\documentclass[english,11pt]{article}
\usepackage{graphicx} 

\usepackage{fullpage}
\usepackage{float}
\parskip=\smallskipamount
\usepackage{booktabs} 

\usepackage{amsmath,amssymb,amsthm,mathtools}
\usepackage{dsfont}
\usepackage{latexsym}
\usepackage[vlined,linesnumbered,ruled]{algorithm2e}
\SetKw{KwReturn}{return}
\usepackage{algpseudocode}
\usepackage{tikz}
\usetikzlibrary{arrows.meta,positioning,calc,decorations.pathreplacing}

\usepackage{array}

\newtheorem{theorem}{Theorem}[section]
\newtheorem{lemma}[theorem]{Lemma}

\newtheorem{fact}[theorem]{Fact}
\newtheorem{observation}[theorem]{Observation}

\newtheorem{corollary}[theorem]{Corollary}
\newtheorem{remark}[theorem]{Remark}
\def\LC{\mbox{\sf LC}}
\def\QA{\mbox{\sf QA}}
\def\CA{\mbox{\sf CA}}
\def\DRC{\mbox{\sf Color\_All}}
\def\NBC{\mbox{\sf NBC}}

\def\ILDRC{\mbox{\sf Iter\_NBC}}
\def\Split{\mbox{\sf Split}}
\def\Recolor{\mbox{\sf Recolor}}

\def\OpenRecolor{\mbox{\sf Open\_Recolor}}

\def\BudgetRecolor{\mbox{\sf 
Budget\_Recolor}}
\def\LogCompact{\mbox{\sf LogCap\_Compact}}

\def\PFTradeoff{\mbox{\sf P-F\_Trade}}
\def\PTTradeoff{\mbox{\sf P-T\_Trade}}
\newcommand{\SmallVariant}[1]{#1_{\mathrm{SMALL}}}
\def\PFTradeoffSmall{\SmallVariant{\PFTradeoff}}
\def\PTTradeoffSmall{\SmallVariant{\PTTradeoff}}
\def\Kierstead{\mbox{\sf EquiColor}}

\def\HIST{\mbox{\sf Hist}}
\def\cP{\mathcal{P}}

\def\Tseq{T^{seq}}
\def\Tcongest{T^{cngst}}
\def\Tcc{T^{cc}}

\def\LARG{\cP_{large}}

\def\CONGEST{\mbox{\tt CONGEST}}
\def\LOCAL{\mbox{\tt LOCAL}}
\def\CC{\mbox{\tt CC}}

\def\clr{\phi}
\def\numclr{\chi}
\def\freq{{\tt f}}
\def\tfreq{{\tt f}} 
\def\Fsmall{\freq_{min}}
\def\Flarge{\freq_{max}}
\def\Fcap{\freq_{cap}}
\def\Fadd{\freq_{add}}
\def\Fpre{\Flarge'}

\newcommand{\du}{d_{unc}}  

\usepackage{xcolor}
\usepackage{soul}
\soulregister\ref7
\soulregister\eqref7
\soulregister\cite7

\long\def\comment #1\commentend{}
\long\def\commentstart #1\commentend{}

\def\inline#1:{\par\vskip 7pt\noindent{\bf #1:}\hskip 10pt}
\def\midinline#1:{\par\noindent{\bf #1:}\hskip 10pt}
\def\dnsinline#1:{\par\vskip -7pt\noindent{\bf #1:}\hskip 10pt}
\def\ddnsinline#1:{\newline{\bf #1:}\hskip 10pt}
\def\largeinline#1:{\par\vskip 7pt\noindent{\large\bf #1:}\hskip 10pt}

\colorlet{hlorange}{orange!30}

\begin{document}

\title{
Distributed Algorithms for Near-Equitable Coloring
}

\author{Amit Nir\thanks{Weizmann Institute of Science. 
E-mail: {amit.nir,david.peleg@weizmann.ac.il}}
\and
David Peleg$^*$}





\date{\today}

\maketitle

\begin{abstract}
For an $n$-vertex graph of maximum degree $\Delta$ and diameter $D$,  an \emph{equitable} $(\Delta+1)$-coloring is a vertex coloring where the \emph{frequency} of each color (namely, the number of vertices it colors) are all equal to $\sigma=n/(\Delta+1)$ (up to rounding) . The Hajnal-Szemer\'edi Theorem guarantees the existence of such a coloring for every graph, and an $O(n^2\Delta)$ time sequential algorithm is known for computing such a coloring. Here, we study \emph{near-equitable} graph coloring in distributed networks.
The main question of interest is how close one can remain to the desired palette size of $\Delta+1$ while computing, in few distributed rounds, a coloring whose frequencies are close to $\sigma$.
It appears that these two conflicting parameters exhibit a tradeoff, which we attempt to explore.
We present a suite of fast randomized distributed algorithms representing varying points on this tradeoff, analyze their properties, and study their time complexity in the sequential,
$\CONGEST$
and Congested Clique ($\CC$) models. 

\end{abstract}

\section{Introduction}

\subsection{Background and motivation}

Distributed graph coloring is a fundamental tool for decentralized coordination of parallel activities, with a variety of applications. As each color class forms an independent set of vertices in the graph, the coloring partitions the graph into sets of mutually nonadjacent processors or resources. This can be used for symmetry breaking, conflict-free resource allocation, maximizing parallel processing (while preventing race conditions or deadlocks between concurrently operating vertices), and decentralized scheduling. In scheduling applications, the conflict relation between jobs can be represented by a graph: vertices are jobs, and edges connect jobs that cannot be executed in the same time slot. A proper coloring then gives a feasible schedule, with each color corresponding to one slot and each color class to a set of mutually compatible jobs. This graph-coloring viewpoint is classical in timetabling and examination scheduling~\cite{Wood1969timetabling,DeWerra1985timetabling}, and also appears in processor scheduling with mutual exclusion constraints~\cite{BakerCoffman1996mutual}.

For certain applications, it may be desirable to use an \emph{equitable} coloring, namely, a coloring with equal-size (up to $\pm 1$) color classes. For example, 
%
in scheduling applications as above, 
a lower bound on 
the size of a color class prevents underused slots, while an upper bound prevents
overloaded ones; in bounded coloring and mutual-exclusion scheduling, this upper
capacity is imposed explicitly.  Similarly, in parallel graph computations,
skewed color classes lead to load imbalance and inefficient resource utilization.
Imposing a two-sided requirement is also common in load-balanced uses of coloring.

The celebrated Hajnal-Szemer\'edi Theorem~\cite{HajnalSzemeredi1970} guarantees the existence of an equitable coloring with $\Delta+1$ colors, where $\Delta$ is the maximum vertex degree. 
Currently, the fastest constructive algorithm for generating such a coloring is that of Kierstead et al. ~\cite{Kierstead2010fast}, which finds an equitable $\numclr$-coloring in $O(\numclr n^2)$ time when applied with parameter $\numclr$; for $\Delta+1$-coloring, this gives $O(\Delta n^2)$ sequential time. This algorithm can be transformed in the standard manner\footnote
{We omit a description of these distributed implementations.} into a distributed algorithm in the $\CONGEST$ or Congested Clique (hereafter, $\CC$) models (cf.~~\cite{Peleg00:book,LPPP03}).
Note that in contrast to coloring, which is essentially a local property, equitability is a global property in nature, which makes it apparently difficult to achieve by purely local algorithms with high probability\footnote{In this paper we do not consider bounds that hold in expectation.}.
The following theorem serves as a baseline for our later comparisons. Throughout, we denote sequential time complexity by $\Tseq$,
time complexity in the $\CONGEST$ model by $\Tcongest$, and time complexity in the Congested Clique ($\CC$) model by $\Tcc$.

\begin{theorem}
{\bf (Baseline: Equitable coloring)}
\label{thm:baseline-kierstead}
Algorithm $\Kierstead$ of Kierstead et al.~\cite{Kierstead2010fast} computes a valid equitable coloring, with palette size $\numclr = \Delta+1$ and frequency thresholds $[\lfloor\sigma\rfloor,\lceil\sigma\rceil]$, where $\sigma=n/(\Delta+1)$.
Its running times\footnotemark[2] are 
$\Tseq(\Kierstead)=O(\Delta n^2)$,
$\Tcongest(\Kierstead)=O(m)$ and
$\Tcc(\Kierstead)=O(\Delta)$.
\end{theorem}

As faster methods for achieving precise equity might be difficult, we introduce the more relaxed notion of \emph{near-equitable coloring}, and
study algorithms for achieving near-equity, either by relaxing the constraint on the palette size and allowing slightly more than $\Delta+1$ colors or by allowing some slack in the equitability requirement (or both).
For clarity, we first describe our algorithms in an abstract high-level code, and defer
the discussion on their implementation in the sequential setting and in the $\CONGEST$ and Congested Clique ($\CC$) distributed models, and the corresponding time analysis, to a later stage (Section \ref{sec:congest-model}). 



Formally, we consider a \emph{coloring function} $\clr: V \mapsto \cP$, assigning each vertex $v\in V$ a color from the palette $\cP$ such that no two neighbors share the same color. The \emph{palette size} is the number of colors used, denoted $\numclr=|\cP|$.
Hereafter, whenever using a partial palette marked $\cP_s$ (for some descriptor $s$), we will freely use the notation $\numclr_s$ for its size $|\cP_s|$, without explicitly defining it.

The \emph{color class} $V(i)=\{v \mid \clr(v)=i\}$ of color $i$ is the set of vertices colored by $i$, and the \emph{frequency} of $i$ is the number of vertices colored by $i$, $\freq(i) = |V(i)|$.
For a set $S$ of colors, the \emph{total frequency} of $S$, denoted $\tfreq(S) = \sum_{i\in S} \freq(i)$, is the total number of vertices assigned to those colors.
Throughout, we use $\sigma=\frac{n}{\Delta+1}$ to denote the ``idea'' size of a color class. In an equitable coloring, each color class should be of size $\lfloor\sigma\rfloor$ or $\lceil\sigma\rceil$.

A near-equitable coloring allows some flexibility: given a predefined \emph{frequency range} $[\Fsmall, \Flarge]$, where $\Fsmall$ and $\Flarge$ are the (integral\footnote{Hereafter, whenever we write $[a,b]$ for non-integral $a,b$, the intended interval is $[\lfloor a\rfloor,~\lceil b\rceil]$.}) \emph{frequency threshold} parameters\footnote{Different algorithms in this paper utilize different $\Fsmall$ and $\Flarge$ values.}, the frequency of each color $c$ must satisfy $\freq(c) \in [\Fsmall,\Flarge]$.
(For a strict equitable coloring, the frequency range is $[\lfloor\sigma\rfloor, \lceil\sigma\rceil]$.)

\begin{table}[!t]
\centering
\footnotesize
\renewcommand{\arraystretch}{1.18}
\setlength{\tabcolsep}{2pt}
\resizebox{\textwidth}{!}{%
\begin{tabular}{|>{\raggedright\arraybackslash}p{0.16\textwidth}|>{\raggedright\arraybackslash}p{0.36\textwidth}|>{\raggedright\arraybackslash}p{0.64\textwidth}|>{\centering\arraybackslash}p{0.08\textwidth}|}
\hline
\textbf{Algorithm}  & \textbf{Guarantee} & \textbf{Running time} & \textbf{Src}
\\ 
\hline
Baseline 1: Distributed Coloring
&
\begin{tabular}[t]{@{}l@{}}
Palette: $\Delta+1$\\
Freq: $[0,n]$
\end{tabular}
&
\begin{tabular}[t]{@{}l@{}}
Seq: $O(m)$, det.\\
$\CONGEST$: $O(\lg^5\lg n)$, rand., w.h.p.\\
$\CC$: $O(1)$, det.
\end{tabular}
& 
\begin{tabular}[t]{@{}l@{}}
cf. \cite{Johansson-IPL-99} \\
\cite{CzumajDaviesParter2020,CzumajDaviesParter2021} \\
\end{tabular}
\\ 
\hline
Baseline 2: $\Kierstead$ 
&
\begin{tabular}[t]{@{}l@{}}
Palette: $\Delta+1$\\
Freq: $[\lfloor\sigma\rfloor,\lceil\sigma\rceil]$
\end{tabular}
&
\begin{tabular}[t]{@{}l@{}}
Seq: $O(n^2\Delta)$, det.\\
$\CONGEST$: $O(m)$, det.\\
$\CC$: $O(\Delta)$, det.
\end{tabular}
& 
\cite{Kierstead2010fast}
\\
\hline
Near-Balanced Coloring, $\NBC$
&
\begin{tabular}[t]{@{}l@{}}
Palette: $[\frac{\Delta+1}{2\epsilon},\frac{\Delta+1}{\epsilon}]$\\
Freq: $[\epsilon\sigma,2\epsilon\sigma]$
\end{tabular}
&
\begin{tabular}[t]{@{}l@{}}
Seq: $O(m)$, det.\\
$\CONGEST$: $O(D+\Delta+\lg^5\lg n)$, rand., w.h.p.\\
$\CC$: $O(1)$, det.
\end{tabular}
&
Sec. \ref{s: longn time 3Delta palette [sigma/3, 2sigma/3]}
\\ 
\hbox{\hskip 30pt}
Example 
with $\epsilon=1/3$
&
\begin{tabular}[t]{@{}l@{}}
Palette: $[\frac{3(\Delta+1)}{2},3(\Delta+1)]$\\
Freq: $[\sigma/3,2\sigma/3]$
\end{tabular}
&
\begin{tabular}[t]{@{}l@{}}
Seq: $O(m)$, det.\\
$\CONGEST$: $O(D+\Delta+\lg^5\lg n)$, rand., w.h.p.\\
$\CC$: $O(1)$, det.
\end{tabular}
&
\\ 
\hline
Iter. Near-Bal. Coloring, $\ILDRC$
&
\begin{tabular}[t]{@{}l@{}}
Palette: $[\Delta+1,2(\Delta+1)]$\\
Freq: $[\sigma/2,\sigma]$
\end{tabular}
&
\begin{tabular}[t]{@{}l@{}}
Seq: $O(m\lg\Delta)$, det.\\
$\CONGEST$: $O(\lg\Delta\cdot(D+\Delta+\lg^5\lg n))$, rand., w.h.p.\\
$\CC$: $O(\lg\Delta)$, det.
\end{tabular}
&
Sec. \ref{s: ILDRC}
\\ 
\hline
\begin{tabular}[t]{@{}l@{}}
Palette-Freq.\\
Tradeoff\\
\scalebox{0.85}{$\PFTradeoffSmall$}
\end{tabular}
&
\begin{tabular}[t]{@{}l@{}}
Palette: $(1+1/k)(\Delta+1)$\\
Freq: $[\sigma/3,2k\sigma]$
\end{tabular}
&
\begin{tabular}[t]{@{}l@{}}
Seq: $O(m)$, det.\\
$\CONGEST$: $O(D+\Delta+\lg^5\lg n)$, rand., w.h.p.\\
$\CC$: $O(1)$, det.
\end{tabular}
&
Sec. \ref{sec:generalized-coloring}
\\ 
\begin{tabular}[t]{@{}l@{}}
Palette-Freq.\\
Tradeoff\\
$\PFTradeoff$
\end{tabular}
& 
\begin{tabular}[t]{@{}l@{}}
Palette: $(1+1/k)(\Delta+1)$\\
Freq: $[\sigma/2,2k\sigma]$
\end{tabular}
&
\begin{tabular}[t]{@{}l@{}}
Seq: $O(m\lg\Delta)$, det.\\
$\CONGEST$: $O(\lg\Delta\cdot(D+\Delta+\lg^5\lg n))$, rand., w.h.p.\\
$\CC$: $O(\lg\Delta)$, det.
\end{tabular}
&
\\ 
\hline
\begin{tabular}[t]{@{}l@{}}
Palette-Time\\
Tradeoff\\
\scalebox{0.85}{$\PTTradeoffSmall$}
\end{tabular}
&
\begin{tabular}[t]{@{}l@{}}
Palette: $(1+1/k)(\Delta+1)$\\
Freq: $[\sigma/3,2\sigma]$
\end{tabular}
&
\begin{tabular}[t]{@{}l@{}}
Seq: $O(mk)$, det.\\
$\CONGEST$: $O(k\cdot(D+\Delta+\lg^5\lg n))$, rand., w.h.p.\\
$\CC$: $O(k)$, det.
\end{tabular}
&
Sec. \ref{sec:fast-converging}
\\ 
\begin{tabular}[t]{@{}l@{}}
Palette-Time\\
Tradeoff\\
$\PTTradeoff$
\end{tabular}
& 
\begin{tabular}[t]{@{}l@{}}
Palette: $(1+1/k)(\Delta+1)$\\
Freq: $[\sigma/2,2\sigma]$
\end{tabular}
&
\begin{tabular}[t]{@{}l@{}}
Seq: $O(m(k+\lg\Delta))$, det.\\
$\CONGEST$: $O((\lg\Delta+k)\cdot(D+\Delta+\lg^5\lg n))$, rand., w.h.p.\\
$\CC$: $O(k+\lg\Delta)$, det.
\end{tabular}
&
\\ 
\hline
Capped Coloring, $\OpenRecolor$
&
\begin{tabular}[t]{@{}l@{}}
Palette: $(1+\rho(\alpha))\cdot(\Delta+1)+1$ 
\\
for $\rho(\alpha)=(1+2\epsilon)/(2\alpha)$
\\
Freq: $[\sigma/2,\alpha\sigma]$ ($\alpha\ge1$)
\end{tabular}
&
\begin{tabular}[t]{@{}l@{}}
Seq: $O(m\lg n)$\\
$\CONGEST$: $O((D+\Delta)\lg n+\lg\Delta\cdot\lg^5\lg n)$\\
$\CC$: $O(\lg n)$\\
All: rand., w.h.p.
\end{tabular}
&
Sec. \ref{sec:improved_coloring}
\\ 
Example
with $\alpha=1$
&
\begin{tabular}[t]{@{}l@{}}
Palette: $(1.5+\epsilon)(\Delta+1)$\\
Freq: $[\sigma/2,\sigma]$
\end{tabular}
&
\begin{tabular}[t]{@{}l@{}}
Seq: $O(m\lg n)$\\
$\CONGEST$: $O((D+\Delta)\lg n+\lg\Delta\cdot\lg^5\lg n)$\\
$\CC$: $O(\lg n)$\\
All: rand., w.h.p.
\end{tabular}
&
\\ 
\hline
\begin{tabular}[t]{@{}l@{}}
Optimal\\
palette,\\
\textsf{LogCap\_}\\
\textsf{Compact}
\end{tabular}
&
\begin{tabular}[t]{@{}l@{}}
Palette: $\Delta+1$\\
Freq: $[\sigma/2,2\sigma+\lg(\Delta+1)$\\
\hphantom{Freq: }$\cdot(1+\varepsilon)\sigma]$\\
for $\sigma=\Omega(\lg n)$
\end{tabular}
&
\begin{tabular}[t]{@{}l@{}}
Seq: $O(m\lg n\lg\Delta)$\\
$\CONGEST$: $O((D+\Delta)\lg n\lg\Delta)$\\
$\CC$: $O(\lg n\lg\Delta)$\\
All: rand., w.h.p.
\end{tabular}
&
Sec. \ref{sec:delta-plus-two-log-cap}
\\ 
\hline
\end{tabular}
}
\caption{Comparison of near equitable coloring algorithms and their running times in the $\CC$, $\CONGEST$, and sequential models (ceilings/floors omitted for readability).}
\label{tbl:results summary}
\end{table}

\subsection{Our contributions}


The central question studied in this work concerns the potential tradeoffs between the \emph{palette size} $\numclr$ and the \emph{frequency range} $[\Fsmall, \Flarge]$ that can be achieved by randomized distributed algorithms.
Our main results, summarized in Table \ref{tbl:results summary}, sketch the trade-offs available for distributed equitable coloring, ranging from time-costly perfectly equitable solutions to  time-efficient 
approximations.
%

Hereafter, for a given graph $G$, we denote by $n, m, \Delta, D$ its number of vertices, number of edges, maximum degree and diameter, respectively. 




\subsubsection{Basic tools}


The algorithms described later on are assembled from a small toolbox 
consisting of a number of primitives and procedures, whose implementation depends on the computational model.
These 
are the following:

\noindent\textbf{Color accounting ($\CA$).}
This primitive computes aggregate information per color and disseminates the result.  Its uses include, for example, computing the frequencies of colors, sorting colors by frequencies, identifying colors whose frequency is too small or too large,
informing vertices that need to recolor themselves (and by which color classes), and maintaining various global counters on the colors.

\noindent\textbf{Quota assignment ($\QA$).}
This primitive receives, for some color $i$, a set $C_i$ of vertices that request to be colored by $i$ and a \emph{quota} $b_i$, and admits exactly $\min\{|C_i|,b_i\}$ candidates. The other candidates are informed that their request was denied.

\noindent\textbf{Procedure $\DRC(G,\numclr)$:} This procedure is used by our algorithms for the basic initialization of the coloring. It invokes a standard list-coloring procedure $\LC$ with target set $V$ and color lists $L(v)=\{1,\ldots,\numclr\}$ for every vertex. 
The resulting coloring $\clr$ makes no guarantees on color frequencies.

\noindent\textbf{Procedure $\Recolor(V_{target},\cP_{target})$:} 
This procedure is used for modifying an existing coloring. It is typically invoked after some colors were disqualified (say, due to having too small frequencies), and the corresponding vertices were uncolored.
These vertices form the set $V_{target}$ that needs to be recolored.
$\cP_{target}$ is the set of colors to be used for these vertices.



%
%
%
%
%
%
%

For ease of comparison, the main result statements below include the sequential, $\CONGEST$ and Congested Clique running times.  
The model-dependent implementation details leading to these time bounds are deferred to Section~\ref{sec:congest-model}.
\subsubsection{Near-Balanced Coloring ($\NBC$)}


In Section~\ref{s: longn time 3Delta palette [sigma/3, 2sigma/3]}, we present a basic $O(1)$ $\CC$-time Near-Balanced Coloring algorithm ($\NBC$) that produces a coloring with frequency range $[\varepsilon\sigma, 2\varepsilon\sigma]$ for $\varepsilon \le 1/3$.

The algorithm employs a basic \emph{one-shot split \& recolor} paradigm: It generates an initial valid coloring, applies Procedure $\Split$ to eliminate large-frequency colors, then applies Procedure $\Recolor$ to eliminate small-frequency colors, and finally invokes Procedure $\Split$ again to get rid of any newly generated large-frequency colors.

The simplicity of this algorithm hinges on two facts. First, the considered frequency ranges are \emph{doubling}, i.e., $\Flarge \ge 2\Fsmall$, making it simpler to use Procedure $\Split$ (Section 
\ref{sec:split}).
Second, its upper frequency threshold is \emph{smaller than $\sigma$}, which necessitates using more than $\Delta+1$ colors and makes Procedure $\Recolor$ 
easier to apply.
We get the following result.

\begin{theorem}
\label{thm:summary-nbc-intro}
For every $\varepsilon\le 1/3$, Algorithm $\NBC$ computes a valid coloring with palette size $\numclr \in [(\Delta+1)/2\varepsilon,(\Delta+1) / \varepsilon]$ 
and frequency thresholds $[\varepsilon\sigma, 2\varepsilon\sigma]$.
For fixed $\varepsilon$, its running times are
$\Tseq(\NBC)=O(m)$ and $\Tcc(\NBC)=O(1)$. In the $\CONGEST$ model, the coloring is obtained in time  $\Tcongest(\NBC)=O(D+\Delta+\lg^5\lg n)$ with high probability.
\end{theorem}

\subsubsection{Iterative Near-Balanced Coloring ($\ILDRC$)}

We next turn to examining frequency ranges with slightly larger (yet still doubling) thresholds, particularly $\Fsmall=\lfloor\sigma/2\rfloor$ and $\Flarge = \lceil\sigma\rceil$. 
For this purpose, 
we make use of an \emph{iterative} split \& recolor paradigm, where the splitting and recoloring steps are repeated a number of times.
The range is still doubling, so splitting remains safe, but it is closer to the ideal frequency $\sigma$ and therefore leaves less slack in the number of available target colors.

Specifically, Section~\ref{s: ILDRC} introduces the \emph{Iterative Near-Balanced Coloring} ($\ILDRC$) algorithm, which applies procedures $\Recolor$ and $\Split$ repeatedly for $O(\lg\Delta)$ iterations.
The algorithm uses the same initial coloring algorithm $\DRC$.
Then, each iteration consists of invoking Procedure $\Split$ with threshold $\Flarge=\sigma$, so no current color is too large, followed by Procedure $\Recolor$, which uses the $\Delta+1$ most frequent colors as a target palette, and recolors all remaining colors.

Each individual iteration suffers from the same basic difficulty as before: 
After each recoloring step, the next split repairs the upper bound; after each split, the next recoloring step again eliminates the less frequent colors. Hence again, the concern is that the introduction of new colors in the splitting phases might undo the progress made by the recoloring phases. 

The analysis shows that this does not happen. In each iteration, the split step creates no new small colors, and the subsequent recoloring step eliminates \emph{more than half} of the small colors present after the split. Thus, the number of small colors decreases geometrically across the iterations.

At the cost of increasing the $\CC$ runtime to $O(\lg\Delta)$, this algorithm achieves a frequency range of $[\sigma/2, \sigma]$ with palette size $\numclr \le 2(\Delta+1)$.
We get the following result.


\begin{theorem}
\label{thm:summary-ildrc-intro}
Algorithm $\ILDRC$ computes a valid coloring with palette size $\numclr\in [\Delta+1, 2(\Delta+1)]$ 
and frequency thresholds $[\sigma/2,\sigma]$.
Its running times are $\Tseq(\ILDRC)=O(m\lg\Delta)$ and
$\Tcc(\ILDRC)=O(\lg\Delta)$.
In the $\CONGEST$ model, the coloring is obtained in time  $\Tcongest(\ILDRC)=O(\lg\Delta\cdot(D+\Delta+\lg^5\lg n))$ with high probability.
\end{theorem}

\subsubsection{Palette - Frequencies Tradeoff ($\PFTradeoff$)}

Next, Section~\ref{sec:generalized-coloring} presents a generalized algorithm that, for any parameter $k \ge 1$, achieves a palette size of $(1+1/k)(\Delta+1)$ with frequency range $[\sigma/2, 2k\sigma]$. To achieve this, we revisit the simpler \emph{one-shot} split \& recolor paradigm, but introduce two modifications. First, we ``bootstrap'' the process by using Algorithm $\ILDRC$ of Section~\ref{s: ILDRC} to generate the initial coloring. Secondly, we use a much larger upper frequency threshold of $\Flarge=2k\sigma$ for the final splitting step, thus reducing the need for splitting and allowing us to end up with fewer colors.

In the recoloring stage, the algorithm uses the $\Delta+1$ most frequent colors to recolor all remaining vertices.  This may create large color classes, so the final step applies $\Split$ with the relaxed upper
threshold $2k\sigma$. Thus, the parameter $k$ controls the trade-off: a larger $k$ allows larger frequencies and yields a palette closer to $\Delta+1$.


\begin{theorem}
\label{thm:summary-tradeoff-intro}
For every integer parameter $k\ge1$, Algorithm $\PFTradeoff$ computes a valid coloring with palette size $\numclr \le (1+1/k) (\Delta+1)$ and frequency thresholds $[\sigma/2,2k\sigma]$.
Its running times are $\Tseq(\PFTradeoff)=O(m\lg\Delta)$ and $\Tcc(\PFTradeoff) = O(\lg\Delta)$.
In the $\CONGEST$ model, the coloring is obtained in time $\Tcongest(\PFTradeoff) = O(\lg\Delta\cdot(D+\Delta+\lg^5\lg n))$ with high probability.

\end{theorem}
\begin{theorem}
\label{thm:summary-tradeoff-small-intro}
For every integer parameter $k\ge1$, Algorithm $\PFTradeoffSmall$ computes a valid coloring with palette size $\numclr\le(1+1/k)(\Delta+1)$ and frequency thresholds $[\sigma/3,2k\sigma]$.
Its running times are $\Tseq(\PFTradeoffSmall)=O(m)$ and $\Tcc(\PFTradeoffSmall)=O(1)$.
In the $\CONGEST$ model, the coloring is obtained in time $\Tcongest(\PFTradeoffSmall)=O(D+\Delta+\lg^5\lg n)$ with high probability.
\end{theorem}

\subsubsection{Palette - Time Tradeoff ($\PTTradeoff$)}

In Section~\ref{sec:fast-converging}, we further optimize the convergence rate, achieving a palette of $(1+1/k)(\Delta+1)$ with tighter frequency range $[\sigma/2, 2\sigma]$, but at the cost of increasing the $\CC$ time complexity to $O(k+\lg\Delta)$.
So, the resulting algorithm, $\PTTradeoff$, yields
a trade-off complementary to that of Algorithm $\PFTradeoff$ of Section \ref{sec:generalized-coloring}.  The algorithm is based on the iterative split \& recolor paradigm, using Algorithm $\ILDRC$ for the initial coloring and continuing the iterations until the number of colors drops below the desired bound. Instead of allowing the upper frequency threshold to grow with $k$, the algorithm keeps the tighter frequency range $[\sigma/2, 2\sigma]$. The price is increased time: it repeatedly compacts the palette to the largest $\Delta+1$ colors and then splits colors above $2\sigma$, so the excess palette size decreases by about $(\Delta+1) /k$ in each iteration.


\begin{theorem}
\label{thm:summary-fast-converge-intro}
For every integer parameter $k\ge1$, $\PTTradeoff$ (Algorithm~\ref{alg:fast-converge}) computes a valid coloring with palette size $\numclr\le(1+1/k)(\Delta+1)$ and frequency thresholds $[\sigma/2,2\sigma]$.
Its running times are $\Tseq(\PTTradeoff) = O(m(k+\lg\Delta))$ and $\Tcc(\PTTradeoff) = O(k+\lg\Delta)$. In the $\CONGEST$ model, the coloring is obtained in time  $\Tcongest(\PTTradeoff) = O((\lg\Delta+k)\cdot(D+\Delta+\lg^5\lg n))$ with high probability.
\end{theorem}

\begin{theorem}
\label{thm:summary-fast-converge-small-intro}
For every integer parameter $k\ge1$, Algorithm $\PTTradeoffSmall$ computes a valid coloring with palette size $\numclr\le(1+1/k)(\Delta+1)$ and frequency thresholds $[\sigma/3,2\sigma]$.
Its running times are $\Tseq(\PTTradeoffSmall)=O(mk)$ and $\Tcc(\PTTradeoffSmall)=O(k)$.
In the $\CONGEST$ model, the coloring is obtained in time $\Tcongest(\PTTradeoffSmall)=O(k\cdot(D+\Delta+\lg^5\lg n))$ with high probability.
\end{theorem}


\subsubsection{Capped / open Reoloring ($\OpenRecolor$)}

Section~\ref{sec:improved_coloring} presents Algorithm $\OpenRecolor$, which refines the iterative approach, reducing the palette size to $(1.5+\epsilon)(\Delta+1)$ while maintaining the $O(\lg n)$ $\CC$ runtime and the frequency range $[\sigma/2, \sigma]$. 
More generally, the algorithm is parameterized by an upper-frequency factor $\alpha\ge1$ (setting $\Flarge=\alpha\sigma$) and an integer \emph{excess} $\ell$ over the optimal palette size $\Delta+1$. We write $\beta=1+\ell/(\Delta+1)$, so the target palette has
$\numclr_{target}=\beta(\Delta+1)=\Delta+1+\ell$ colors, and the resulting frequency range is $[\sigma/2, \lceil\alpha\sigma\rceil]$.

The main idea behind Algorithm $\OpenRecolor$ is the following. Recall that Algorithm $\ILDRC$
(Section \ref{s: ILDRC})
is based on the iterative split \& recolor paradigm, alternating between \emph{recoloring phases}, which reduce the number of colors at the cost of increasing the frequencies of the remaining colors, and \emph{splitting phases}, which increase the number of colors. In contrast, the new approach aims to prevent the need for splitting entirely, thus reducing the number of phases. To achieve that, we no longer treat all colors with frequency at most $\Flarge$ as valid targets in the recoloring stage. Rather, we forbid the use of a color $c$ if its frequency $\freq(c)$ is so 
close to $\Flarge$ that there's a risk that it might exceed $\Flarge$ if $c$ is used in the recoloring stage.

Following this idea, the key new ingredient is to recolor uncolored vertices using the target palette while enforcing the upper threshold deterministically. To this end, we use an \emph{open-cap} parameter $\varphi<1$ to determine which colors are still available during the recoloring phase, and define an intermediate \emph{capping threshold} $\Fcap = \lfloor \varphi \alpha\sigma \rfloor$ in-between $\Fsmall$ and $\Flarge$. A color $i$ is now considered \emph{open} only while its frequency satisfies $\freq(i) \le \Fcap$.
Colors with frequency above $\Fcap$ are declared \emph{closed}. Only \emph{open} colors may be used for coloring vertices at the recoloring stages, and moreover, each open color can be used only up to its remaining capacity, $\Flarge-\Fcap$.

The parameters $\alpha$, $\ell$, $\beta$, $\varphi$ are chosen to satisfy
\[
\beta=1+\frac{\ell}{\Delta+1}\mbox{ for }1\le\ell\le\Delta,\qquad
\alpha\ge1,\qquad
\frac12<\varphi<1,\qquad
2\varphi\alpha\cdot\frac{\ell}{\Delta+1}>1 .
\]

Under these constraints, we have the following result.

\begin{theorem}
\label{thm:summary-open-recolor-intro}
For fixed constant parameters satisfying the conditions above, Algorithm $\OpenRecolor$ computes, with high probability, a valid coloring with palette size $\numclr \le \beta (\Delta+1)$ and frequency thresholds $[\sigma/2, \lceil\alpha\sigma\rceil]$.
Its running times are $\Tseq(\OpenRecolor) = O(m\lg n)$, $\Tcongest(\OpenRecolor) = O((D+\Delta)\lg n+\lg\Delta\cdot\lg^5\lg n)$ and
$\Tcc(\OpenRecolor) = O(\lg n)$.
\end{theorem}

\subsubsection{Optimal $\Delta+1$ palette ($\LogCompact$)}
Section~\ref{sec:delta-plus-two-log-cap} pushes the palette size down to $\Delta+1$ by applying the capped recoloring idea gradually until $\numclr=\Delta+2$, and then with one last $\Recolor$ gets to $\numclr=\Delta+1$.  In each phase, the algorithm removes roughly half of the current excess colors and recolors their vertices into the remaining palette, while allowing every ``surviving'' color (that was not eliminated during the color reduction process) to gain only $O(\sigma)$ new vertices.  Since the number of phases is $O(\lg\Delta)$, the price for the $\Delta+1$ palette is a larger upper frequency threshold of $O(\lg\Delta\cdot\sigma)$.

The external parameters of Algorithm $\LogCompact$ are the \emph{slack} parameter $0<\varepsilon\le1$ and the \emph{failure-probability exponent} $c\ge1$. 
The parameter
$\varepsilon$ controls the amount of extra capacity available in each phase:
a color $i$ that is still in $\cP$ maintains a counter $\Fadd(i)$ for the number of new
vertices assigned to it in the current phase, may receive at most
$\Flarge^{ph}=\lceil(1+\varepsilon)\sigma\rceil$ new vertices, and is open only
while $\Fadd(i)\le \Fcap$, where
$\Fcap=\lfloor(1+\varepsilon/2)\sigma\rfloor$. Thus the gap
$\Flarge^{ph}-\Fcap=\Omega(\varepsilon\sigma)$ serves as a safety buffer 
(to be used for the concentration arguments in the analysis). 
The exponent $c$ sets the final success probability $1-n^{-c}$.
Finally, if a phase has target palette size $\Delta+1+s$, then $s$ denotes the current excess over $\Delta+1$; the phase removes (at most $s$) low-frequency colors from this excess and recolors their vertices using the surviving palette. The analysis establishes that the excess decreases geometrically, so there are $O(\lg\Delta)$ phases, which gives the final upper frequency threshold
$\sigma+\lceil\lg(\Delta+1)\rceil\Flarge^{ph}$.

The parameters $\varepsilon,c$ are chosen to satisfy 
\[
0<\varepsilon\le 1,\qquad 
c\ge 1,\qquad
\sigma \ge \max\{6/\varepsilon,(100/\varepsilon)(c+6)\lg n\}.
\]

\begin{theorem}
\label{thm:summary-delta-plus-two-intro}
For fixed constant parameters satisfying the conditions above, 
Algorithm $\LogCompact$ computes, with probability at least $1-n^{-c}$, a valid coloring with palette size $\numclr=\Delta+1$ and frequency thresholds $[\sigma/2, 2\sigma+\lceil\lg (\Delta+1)\rceil \lceil(1+\varepsilon) \sigma \rceil]$.
Its parameter-dependent running times are $\Tseq(\LogCompact) = O(m\lg n\lg\Delta)$,
$\Tcongest(\LogCompact) = O((D+\Delta)\lg n\lg\Delta)$ and $\Tcc(\LogCompact) = O(\lg n\lg\Delta)$.
\end{theorem}

\section{Basic tools}
\label{sec:tools}

\subsection{Communication and coordination primitives}
\label{sec:CA+QA}

The remaining primitive interfaces are used to coordinate choices indexed by colors.  They do not by themselves recolor vertices; instead, they provide the global accounting and admission steps used by the procedures and algorithms.

\noindent\textbf{Color accounting ($\CA$).}
The color-accounting primitive computes aggregate information per color and disseminates the result.  Its uses include, for example, computing the frequency $\freq(i)$ of a color $i$, sorting the colors by non-increasing order of frequency, identifying the palette $\cP_{remove}$ of colors whose frequency is too small and therefore should be removed, informing all vertices currently using colors from $\cP_{remove}$ that they need to recolor themselves, and announcing to all those vertices 
the target palette $\cP_{target}$ to be used in the recoloring process, testing whether a target set is empty, and maintaining various global counters on the colors.

\noindent\textbf{Quota assignment ($\QA$).}
The quota-assignment primitive receives, for each color $i$, a candidate set $C_i$ of vertices that request to be colored by $i$ and a \emph{quota} $b_i$, and admits exactly $\min\{|C_i|,b_i\}$ candidates. The other candidates are informed that their request was denied.
\\
More specifically, when an algorithm invoking $\QA$ uses random ranks, the primitive $\QA$ admits the highest-ranked feasible candidates.  This interface is used both for capacity-aware recoloring (for example Algorithm~\ref{alg:open-recolor} ($\OpenRecolor$) of section~\ref{sec:hard-capped-open-reduction}) and for assigning vertices of a split color to its descendant colors.

\begin{lemma}
\label{lem:ca-qa-runtime}
Assume that the primitives $\QA$ or $\CA$ are invoked for handling $O(\Delta)$ color records.
Then they
have the following implementation costs:
$\Tseq(\CA)=\Tseq(\QA)=O(m)$, $\Tcongest(\CA)=\Tcongest(\QA)=O(D+\Delta)$ and $\Tcc(\CA)=\Tcc(\QA)=O(1)$.
\end{lemma}

\subsection{List Coloring, $\DRC$ and $\Recolor$}
\label{sec:LC + ColorAll+Recolor}

The first set of key primitives we use concerns coloring. Underlying these operations is the \emph{list coloring} operation.

\noindent\textbf{List Coloring ($\LC$).}
Given a target set $V'\subseteq V$ and, for each $v\in V'$, a list $L(v)$ of admissible colors, $\LC$ returns a valid coloring of $V'$ from these lists, consistent with the colors already assigned to $V\setminus V'$.

We do not develop any new algorithms for $\LC$, but rather use existing $\LC$ algorithms for our needs.
Specifically, we use this primitive to implement two procedures, named $\DRC$ and $\Recolor$, each involving a single application of $\LC$, without using any additional communication primitives.  

\noindent\textbf{Procedure $\DRC(G,\numclr)$:} This procedure is used by our algorithms for the basic initialization of the coloring. It invokes $\LC$ with target set $V$ and color lists $L(v)=\{1,\ldots,\numclr\}$ for every vertex. 
The resulting coloring $\clr$ makes no guarantees on color frequencies.



\begin{lemma}
\label{lem:summary-drc}
Given a palette $\cP$ of size $\numclr\ge\Delta+1$, Procedure $\DRC(G,\numclr)$ produces a valid coloring $\clr:V\to\cP$ using at most $\numclr$ colors.  At the interface level, this is one call to $\LC$. Its running times are $\Tseq(\DRC)=O(m)$ and $\Tcc(\DRC)=O(1)$.
In the $\CONGEST$ model, the coloring is obtained in time  $\Tcongest(\DRC)=O(\lg^5\lg n)$ with high probability.
\end{lemma}
Throughout, the implementation details in the three settings and the analysis of the time bounds stated in our lemmas are deferred to Section~\ref{sec:congest-model}.

\noindent\textbf{Procedure $\Recolor(V_{target},\cP_{target})$:} 
This procedure is used for modifying an existing coloring. It is typically invoked after some colors were disqualified (say, due to having too small frequencies), and the corresponding vertices were uncolored. These vertices form the set $V_{target}$ that needs to be recolored.
$\cP_{target}$ is the set of colors to be used for these vertices.
The procedure extends the current partial coloring by invoking $\LC$ on $V_{target}$, where each vertex deletes from $\cP_{target}$ the colors already used by its colored neighbors. 

\begin{lemma}
\label{lemma:partial_extension}
Let $\clr$ be a valid partial coloring of $G$ with a set $V_{target}$ of uncolored vertices, and let $\cP_{target}$ be a target palette of size $\numclr_{target} = |\cP_{target}| \ge \Delta+1$. Then $\Recolor(V_{target},\cP_{target})$ extends the partial coloring $\clr$ to a valid coloring of $G$ in which every vertex of $V_{target}$ receives a color from $\cP_{target}$.
Letting each vertex learn $\cP_{target}$ requires a $\CA$ operation. Then,
the procedure consists of one call to $\LC$. Its running times are
$\Tseq(\Recolor)=O(m)$ and $\Tcc(\Recolor)=O(1)$.
In the $\CONGEST$ model, the coloring is obtained in time  $\Tcongest(\Recolor)=O(D+\Delta+\lg^5\lg n)$ with high probability.
\end{lemma}

\begin{remark}
Throughout the paper, whenever a color $i$ is dissolved and $V(i)$ becomes empty, we implicitly assume that $i$ is erased from our palette (i.e., we do not write this change explicitly in the code). 
\end{remark}


\begin{remark}
The schematic modular decomposition of our algorithms into basic primitives for the sake of time complexity analysis may be somewhat artificial, especially for the sequential implementation, where color accounting and quota assignments can be performed directly. We use this scheme mainly for simplifying the time analysis of the distributed implementations of our algorithms.
We discuss this in more detail in Section~\ref{sec:congest-model}.
\end{remark}

\subsection{Safe splitting: reducing large color classes}
\label{sec:split}


We now present Procedure $\Split$, used for reducing the color frequencies of a given coloring $\clr$ below a specified upper threshold $\Flarge$,
%
%
by splitting large color classes. This is done by first identifying colors whose frequency exceeds the threshold (using the primitive $\CA$ to count the frequency of each color), and then creating, for each such color $i$, the necessary number of new \emph{descendant} colors of
prescribed (smaller) sizes replacing it, and assigning the vertices of $V(i)$ to these descendant colors (using the primitive $\QA$).


\begin{algorithm}[htb]
\caption{$\Split(\clr, \Flarge)$}
\label{alg:splitting}
\textbf{Input:} A current coloring $\clr: V \to \cP$, frequency threshold $\Flarge \ge 2$. 
\\
\textbf{Output:} A refined coloring $\clr'$ with upper frequency threshold $\Flarge$.


$\clr' \gets \clr$ \\
$\LARG \gets \{ i \in \cP \mid \freq(i) > \Flarge \}$
\\[4pt]
    \For{each color $i \in \LARG$ (in parallel)} {
        Let $\freq(i)=k\cdot\Flarge+t$ for $0\le t<\Flarge$\\
        {\bf If} $t\ge 1$ then
        split $V(i)$ into $r=k+1$ disjoint sets: $V_1,\ldots, V_{k-1}$ of size $\Flarge$, \\
        $V_k$ of size $\lfloor (\Flarge+t)/2 \rfloor$ and $V_{k+1}$ of size $\lceil (\Flarge+t)/2 \rceil$. \\
        {\bf Else} ($t=0$, $k\ge 2$)
        split $V(i)$ into $r=k$ disjoint sets: $V_1,\ldots, V_k$ of size $\Flarge$.\\
        Generate $r$ new unique colors $i_1,\ldots, i_r$.\\
        {\bf for} every $j \in \{1,\ldots,r\}$ {\bf do:}
        set $\clr'(v) \gets i_j$ for every $v \in V_{i_j}$.\\
        $\cP \gets (\cP \setminus \{i\}) \cup \{i_1,\ldots,i_r\}$ \\
    } 
\Return $\clr'$

\end{algorithm}



\begin{lemma}
\label{lemma:safe_splitting_general}
Consider an execution of Procedure $\Split(\clr, \Flarge)$ on a valid coloring $\clr$ of a graph $G$ and a frequency threshold $\Flarge \ge 1$.
\\
(a) $\Split$ returns a valid coloring.
\\
(b) If a color $i$ is split during the execution, 
then any created color $i_{new}$ has frequency 
$\freq(i_{new}) \ge \lceil \Flarge/2 \rceil$.
\\
(c) 
If the input coloring satisfies $\freq(i)\ge\Fsmall$ for every color $i$ and $\Fsmall\le\lfloor\Flarge/2\rfloor$, then every output color satisfies $\freq(i)\in[\Fsmall,\Flarge]$.
\\
(d) The running times of Procedure $\Split$ are $\Tseq(\Split)=O(m), 
\Tcongest(\Split)=O(D+\Delta), 
\Tcc(\Split)=O(1).$
\end{lemma}

\begin{proof}
Each color class $V(i)$ is an independent set, so its partition into 
disjoint sets preserves the independence property. Thus, the resulting coloring is valid, establishing (a).




To prove part (b), consider a color $i$ selected for splitting. By the loop condition, $\freq(i) > \Flarge$, and recalling that $\Flarge$ is integral by convention and $\freq(i)$ is integral by definition, we have $\freq(i) \ge \Flarge+1$.
The algorithm splits 
the color $i$ into new colors, where each new color $i_{new}$ has frequency
\begin{equation}
\label{eq: freq i new}
\freq(i_{new}) \ge \min\{\Flarge,
\left\lfloor (\Flarge+1)/2 \right\rfloor\} \ge \left\lceil \frac{\Flarge}{2} \right\rceil~,
\end{equation}
where the final equality follows since $\Flarge$ is integral.

Claim (c) follows because unsplit colors preserve the input lower bound, and created colors $i_{new}$ have frequency $\freq(i_{new}) \ge \lceil\Flarge/2\rceil\ge\Fsmall$, 
and the procedure terminates only after all frequencies are at most $\Flarge$. 
The runtime bounds in part (d) follow from part (c) and Lemma~\ref{lem:ca-qa-runtime}: $\Split$ performs one $\CA$ step to identify the colors to split and compute the descendant quotas, and one $\QA$ step to assign vertices to the descendant colors; the model-specific implementations of these primitives are deferred to Section~\ref{sec:congest-model}.
\end{proof}

\section{Small Frequency Range Near Balanced Coloring}
\label{sec:small-frequency-range-coloring}
\label{s: longn time 3Delta palette [sigma/3, 2sigma/3]}
In this section, we present and analyze a distributed algorithm that achieves a near-equitable coloring. 
This algorithm employs a basic \emph{one-shot split \& recolor} paradigm: It generates an initial valid coloring, applies Procedure $\Split$ to eliminate large-frequency colors, then applies Procedure $\Recolor$ to eliminate small-frequency colors, and finally invokes Procedure $\Split$ again to get rid of any newly generated large-frequency colors.

While several of our later results aim at using at most $2(\Delta+1)$ 
colors with lower frequency threshold $\sigma/2$, this section aims at generating a coloring with many colors and a lower frequency threshold of $\varepsilon\sigma$ for small $\varepsilon$.

%
%
We classify colors into three categories based on their frequency $\freq(i)$ relative to the algorithm's specific frequency thresholds $\Fsmall$ and $\Flarge$.
\begin{itemize}
\item \textbf{Small:} 
The color $i$ is underpopulated, with $\freq(i) < \Fsmall$. Each small color $i$ is a candidate for dissolution (by recoloring the vertices in $V(i)$).
\item \textbf{Good:} 
The color $i$ is balanced, with $\freq(i)\in[\Fsmall,\Flarge]$. This is the desired state.
\item \textbf{Large:} 
The color $i$ is overpopulated, with $\freq(i) > \Flarge$. These colors must be split to reduce their frequency.
\end{itemize}
We denote the corresponding color palettes by:
\begin{align}
\label{eq: color classes}
\cP_{small} &= \{ i \in \cP \mid \freq(i) < \Fsmall \},
\\
\cP_{good} &= \{ i \in \cP \mid \freq(i)\in[\Fsmall,\Flarge] \},
\nonumber
\\
\cP_{large} &= \{ i \in \cP \mid \freq(i) > \Flarge \}.
\nonumber
\end{align}


Let us start with describing the main idea for the first near-equitable coloring algorithm. Fix a constant $\varepsilon=1/M$
for an integer $M\ge3$, and set the target frequency range to $[\Fsmall,\Flarge]=[\varepsilon\sigma,2\varepsilon\sigma]$. Since the upper threshold is below $\sigma$ when $\varepsilon\le1/3$, a coloring in this range must use more than $\Delta+1$ colors. The advantage is that this additional palette slack makes it possible to apply the one-shot split \& recolor paradigm directly.

The algorithm first runs $\DRC$, which gives a valid coloring but does not bound the color frequencies. It then applies $\Split$ to every color whose frequency is larger than $\Flarge$. After this step, no color is large: each color is either 
good, or is small and should be dissolved.

The next step is to recolor all vertices that currently belong to small colors, using only the good colors as the target palette. A natural concern is that after removing the small colors, too few colors might remain available for applying $\Recolor$. This is exactly where the choice $\varepsilon\le1/3$ is used: since every good color has frequency at most $2\varepsilon\sigma\le2\sigma/3$, a volume 
argument shows that there are at least $\Delta+1$ good colors (Lemma~\ref{lemma:volume}), and therefore the target palette is large enough for the recoloring step.

This recoloring may increase the frequencies of the good colors and may therefore create new large colors. The final call to $\Split$ repairs this by enforcing the upper threshold $\Flarge$ again. Since the range is doubling, namely $\Flarge \ge 2\Fsmall$, splitting a color whose frequency exceeds $\Flarge$ creates new colors of frequency at least $\Flarge/2=\Fsmall$. Hence the final split restores the upper bound without creating new small colors. Consequently, 
the algorithm returns a valid coloring whose frequency range is $[\varepsilon\sigma,2\varepsilon\sigma]$ and whose palette size is in the range $[(\Delta+1)/(2\varepsilon), (\Delta+1)/\varepsilon]$.

The formal description of Algorithm $\NBC$ is given in Algorithm {\ref{alg:balanced_coloring}} ($\NBC$).

\begin{algorithm}[t]
\caption{$\NBC(G,\varepsilon)$: Distributed Near-Balanced Coloring}
\label{alg:balanced_coloring}
\textbf{Input:} Graph $G=(V,E)$ with maximum degree $\Delta$ and size $n = |V|$, parameter $\varepsilon$. \\
\textbf{Output:} A valid coloring $\clr: V \to \cP$ with frequency range $[\varepsilon\sigma,2\varepsilon\sigma]$

$\Fsmall \gets \varepsilon\sigma$; \quad
$\Flarge \gets 2\varepsilon\sigma$
\Comment{frequency threshold parameters} \\

$\clr \gets$ \Call{$\DRC$}{$G, \Delta+1$} \Comment{Initialization: Start with a valid coloring}
\\
$\clr \gets$ \Call{\Split}{$\clr, \Flarge$} \Comment{Reduce all class frequencies to $\le \Flarge$}
\label{step: alg 2 - split}
\\
$\cP_{small} \gets \{ i \in \cP \mid \freq(i) < \Fsmall \}$; \hbox{\hskip 20pt}
$\cP_{good} \gets \{ i \in \cP \mid \freq(i)\in[\Fsmall,\Flarge] \}$
\label{step: alg 2 - classification}
\\
$V_{small} \gets \bigcup_{i \in \cP_{small}} V(i)$ \Comment{Vertices to be recolored} 
\label{step: alg 2 - dissolving}
\\
\Call{\Recolor}{$V_{small}, \cP_{good}$} \Comment{Recolor $V_{small}$ using palette $\cP_{good}$} \\ 
$\clr \gets$ \Call{\Split}{$\clr, \Flarge$} 
\Comment{Recoloring may have caused colors in $R$ to exceed $\Flarge$} 
\label{step: alg 2 - balancing}
\\
\Return $\clr$
\end{algorithm}

\subsection*{Correctness proof and analysis}

\begin{lemma}
\label{lemma:safe_splitting_first_algorithm}
Applying Procedure $\Split$ with frequency threshold $\Flarge= 2\varepsilon\sigma$ guarantees that any resulting color $i'$ has frequency $\freq(i') \ge \Fsmall$.
Consequently, splitting a large frequency color (in $\cP_{large}$) never creates a small-frequency color (in $\cP_{small}$).
\end{lemma}

\begin{proof}
Applying Lemma \ref{lemma:safe_splitting_general} with $\Flarge = \lceil2\varepsilon\sigma\rceil$, 
a split color $i$ with $\freq(i) > \Flarge$ results in new colors $i'$ satisfying 
$\freq(i') \ge \Flarge/2
\ge 
\varepsilon\sigma \ge \lfloor\varepsilon\sigma\rfloor = \Fsmall$,
so $i'\in\cP_{small}$.
\end{proof}

\begin{lemma}
\label{lemma:volume}
If the initial coloring uses at most $\Delta+1$ colors, then $|\cP_{good}| \ge \Delta + 1$.
\end{lemma}

\begin{proof}
Let $\cP^\circ = \cP \setminus \cP_{small}$ denote the set of all non-small colors before any splitting occurs.
The sum of frequencies in the graph satisfies
\begin{equation}
\label{eq:sum of freq}
\tfreq(\cP_{small}) + \tfreq(\cP^\circ) = n.
\end{equation}

First, we bound $\tfreq(\cP_{small})$. The maximum frequency of any color $i\in \cP_{small}$ satisfies $\freq(i) < \varepsilon\sigma$. Since the initial valid coloring (obtained by Algorithm $\DRC$) uses at most $\Delta+1$ colors, we have $|\cP_{small}| \le \Delta+1$.
Therefore,
\[ \tfreq(\cP_{small}) = \sum_{i \in \cP_{small}} \freq(i) < |\cP_{small}|\cdot\varepsilon\sigma=|\cP_{small}| \cdot \frac{\varepsilon n}{\Delta+1} \le
\varepsilon n,\]
so combined with Eq.~\eqref{eq:sum of freq},
\[ \tfreq(\cP^\circ) = n - \tfreq(\cP_{small}) > (1-\varepsilon)n.\]
The splitting process performed in 
Step \ref{step: alg 2 - balancing}
(defined in Lemma \ref{lemma:safe_splitting_general}) ensures that any color in $\cP^\circ$ is either already in $\cP_{good}$ or is split into colors that end up in $\cP_{good}$. Thus, $\tfreq(\cP_{good}) \ge \tfreq(\cP^\circ) > (1-\varepsilon)n$.

Finally, we lower bound $|\cP_{good}|$. Since every color $i \in \cP_{good}$ satisfies the upper bound $\freq(i) \le 2\varepsilon\sigma= \frac{2\varepsilon n}{\Delta+1}$, the minimum number of colors required to account for $\tfreq(\cP_{good})$ is
\[ |\cP_{good}| \ge \frac{\tfreq(\cP_{good})}{\max_{i \in \cP_{good}} \freq(i)} > \frac{(1-\varepsilon)n}{\frac{2\varepsilon n}{\Delta+1}} =  \frac{1-\varepsilon}{2\varepsilon} (\Delta + 1) \ge \Delta+1,\]
where the last inequality follows from the fact that $\varepsilon =1/M \le 1/3$.
\end{proof}

\begin{lemma}
\label{lemma:good_color_sizes}
Given a graph $G$
and an initial coloring with $\freq(i) \ge \varepsilon\sigma$ for every color $i$, Procedure $\Split$ returns 
a refined coloring with frequency range $[\varepsilon\sigma,2\varepsilon\sigma]$.
\end{lemma}

\begin{proof}
Let $\Fsmall = \lfloor\varepsilon\sigma\rfloor$ and $\Flarge = \lceil2\varepsilon\sigma\rceil$. Note that $\Flarge \ge 2\Fsmall$.
The algorithm proceeds in phases. In each phase, every color $i$ performs the following operations in parallel (as described in Procedure $\Split$ (Algorithm~\ref{alg:splitting})):
In each round, the nodes compute the current frequency $\freq(i)$.
\begin{itemize}
    \item If $\freq(i) \le \Flarge$, the color satisfies the upper bound. Since the initial coloring assumed in the lemma
    guarantees $\freq(i) \ge \Fsmall$ and no action is taken, the bounds hold.
    \item If $\freq(i) > \Flarge$, color $i$ is split: $V(i)$ is partitioned into two disjoint subsets $V(i_1)$ and $V(i_2)$, with frequencies $\freq(i_1) = \lfloor \freq(i)/2 \rfloor$ and $\freq(i_2) = \lceil \freq(i)/2 \rceil$.
\end{itemize}
This process repeats until all colors satisfy $\freq(i) \le \Flarge$.
We proved in \ref{lemma:safe_splitting_first_algorithm} that splitting a color which violates the upper bound never results in a color that violates the lower bound.
Thus, for any resulting new color $i'$, $\freq(i') > \Fsmall$.
This ensures that the lower bound $\freq(i') \ge \varepsilon\sigma$ is maintained throughout the execution.
When the process terminates, all colors satisfy the upper bound $\freq(i) \le 2\varepsilon\sigma$. Thus, the final coloring satisfies $\freq(i)\in[\varepsilon\sigma,2\varepsilon\sigma]$ for every $i$.
%
\end{proof}

For the sake of the time analysis of Section~\ref{sec:congest-model} we note the following:

\begin{fact}
\label{lem:nbc-runtime-interface}
Algorithm $\NBC$ makes one call to $\DRC$, two calls to $\Split$, one $\CA$ step to classify colors into $\cP_{small}$ and $\cP_{good}$, and one call to $\Recolor$.
\end{fact}



\begin{theorem}
\label{thm:summary-nbc}
For a graph $G$ 
and a parameter $\varepsilon\le 1/3$,
Algorithm $\NBC(G,\varepsilon)$ computes a valid coloring $\clr$ with palette size $\numclr\in[(\Delta+1)/(2\varepsilon),(\Delta+1)/\varepsilon]$ and frequency range 
$[\Fsmall,\Flarge] =[\varepsilon\sigma,2\varepsilon\sigma]$.
Its running times are $\Tseq(\NBC)=O(m)$ and $\Tcc(\NBC)=O(1)$.
In the $\CONGEST$ model, the coloring is obtained in time  $\Tcongest(\NBC)=O(D+\Delta+\lg^5\lg n)$ with high probability.
\end{theorem}


\begin{proof}
The initialization
executes Procedure $\DRC$ with a palette of size $\numclr$ $= (\Delta+1)$. 
and produces a valid coloring $\clr_{init}$.
\\
The splitting procedure of Step \ref{step: alg 2 - split}, described in Lemma~\ref{lemma:good_color_sizes}, is applied to all colors in $\clr_{init}$.
Any color $i$ with $\freq(i) > \Flarge$ is split recursively.
At the end of this phase, all colors satisfy $\freq(i) \le \Flarge$.
\\
Step \ref {step: alg 2 - classification} classifies the colors into two sets, $\cP_{small}$ and $\cP_{good}$.
By Lemma~\ref{lemma:volume}, since the total frequency is $n$ and $\tfreq(\cP_{small})$ is bounded (due to the bounded number of initial colors), the number of good colors satisfies $|\cP_{good}|\ge \Delta+1$.
\\
Step \ref{step: alg 2 - dissolving} now dissolves the small colors by applying Procedure $\Recolor$, treating the vertices in $V_{small}$ as uncolored and the set $\cP_{good}$ as the global palette of available colors.
We have a valid partial coloring on $V \setminus V_{small}$ using colors from $\cP_{good}$.
The palette size is $|\cP_{good}| \ge \Delta+1$.
By Lemma~\ref{lemma:partial_extension}, we can extend the coloring to $V_{small}$ using the palette $\cP_{good}$.
After this phase, all vertices are colored with colors from $\cP_{good}$. The set $\cP_{small}$ is empty.
\\
The recoloring in 
Step \ref{step: alg 2 - dissolving}
may have caused some colors $i \in \cP_{good}$ to grow beyond the upper bound $\Flarge$ (as they absorbed nodes from $\cP_{small}$).
Step \ref{step: alg 2 - balancing}
achieves the final balancing by applying Procedure $\Split$ on $\cP_{good}$ one final time.
By Lemma ~\ref{lemma:safe_splitting_first_algorithm}, splitting a color with frequency larger than $\Flarge$ results in new colors strictly above $\Fsmall$ in frequency. Thus, no new Small colors are created.

The final coloring $\clr$ is valid as it is obtained from extending a valid initial coloring and splitting color classes (which are independent sets).
All colors $i \in \cP$ satisfy $\freq(i)\in[\Fsmall,\Flarge]$.
Recall that $n = \sum_{i \in \cP} \freq(i)$.
Since $\freq(i) \ge \Fsmall$ for all $i$, we have $n \ge \numclr \cdot \Fsmall$, so
\[ \numclr \le \frac{n}{\Fsmall} = \frac{n}{\frac{\varepsilon n}{\Delta+1}} = \frac{\Delta+1}{\varepsilon}.\]
Conversely, since $\freq(i) \le \Flarge$ for all $i$, we have $n \le \numclr \cdot \Flarge$, so
\[ \numclr \ge \frac{n}{\Flarge} = \frac{n}{\frac{2\varepsilon n}{\Delta+1}} = \frac{\Delta+1}{2\varepsilon}~.\]
Thus, $\numclr \in [(\Delta+1)/(2\varepsilon), (\Delta+1)/\varepsilon]$. The runtime bounds follow from Fact~\ref{lem:nbc-runtime-interface} and are calculated in Section~\ref{sec:congest-model}. The theorem follows. 
\end{proof}


\begin{corollary}
Taking the 
maximal value $\varepsilon=1/3$, 
we obtain a coloring with
frequency range $[\frac{\sigma}{3}, \frac{2\sigma}{3}]$,
palette size
$\numclr$ in
$[1.5(\Delta+1), 3(\Delta+1)]$.
\end{corollary}



\section{Iterated near-equitable coloring algorithm ($\ILDRC$)}
\label{s: ILDRC}

The previous section applied the one-shot split \& recolor paradigm to handle small ranges of the form $[\varepsilon\sigma, 2\varepsilon\sigma]$, where the upper threshold was below $\sigma$. 
In this section, we aim to reach the higher frequency range $[\Fsmall,\Flarge]=[\sigma/2,\sigma]$. For this purpose, one activation of $\Split$ and $\Recolor$ no longer suffices. Instead, we make use of an \emph{iterative} split \& recolor paradigm, where the splitting and recoloring steps are repeated for $O(\lg\Delta)$ times.

The algorithm uses the same initial coloring algorithm $\DRC$.
Then, each iteration consists of invoking Procedure $\Split$ with threshold $\Flarge=\sigma$, so no current color is too large, followed by Procedure $\Recolor$, which uses the $\Delta+1$ most frequent colors as a target palette, and recolors all remaining colors.
%
%
The analysis shows that the split steps do not create new small colors, and in each iteration, the recoloring step eliminates more than half of the remaining small colors. Thus, the number of small colors decreases geometrically with the iterations.
%
%
A final call to $\Split$ enforces the upper bound one last time. Thus, 
the algorithm returns a valid coloring with at most $2(\Delta+1)$ colors, and every color has frequency in $[\sigma/2,\sigma]$.

The formal description of Algorithm $\ILDRC$ is given in Algorithm {\ref{alg:iterative_balanced_final}} ($\ILDRC$).

\begin{algorithm}[htbp]
\caption{Iterative 
Distributed Near Balanced Coloring
\\
Procedure $\ILDRC(G)$}
\label{alg:iterative_balanced_final}
\textbf{Input:} Graph $G=(V,E)$ of size $n$ and Maximum Degree $\Delta$. \\
\textbf{Output:} A valid coloring with $\numclr \le 2(\Delta+1)$ and frequency range 
$[\sigma/2, \sigma]$. 

$\Fsmall \gets \sigma/2$; \quad
$\Flarge \gets \sigma$\\
$\clr \leftarrow$ \Call{$\DRC$}{$G, \Delta+1$}
\Comment{Obtain initial valid $\Delta+1$ coloring}

\For{$t = 1$ to $O(\lg\Delta)$ (sequentially)}
    {
    $\clr \leftarrow$ \Call{\Split}{$\clr, \Flarge$}

    Sort colors in $\cP$ by nonincreasing order of frequency $\freq(i)$: $i_1, i_2, \dots, i_{\numclr}$. 
    \label{step: alg 3 - classification-start}
    \\
    $\cP_{target} \leftarrow \{i_1, \dots, i_{\Delta+1}\}$ \Comment{Top $\Delta+1$ most frequent colors} 
    \label{step: alg 3 - classification-end}
    \\
    $V_{target} \leftarrow \bigcup_{i \in \cP \setminus \cP_{target}} V(i)$ \\
    \Call{\Recolor}{$V_{target}, \cP_{target}$} \Comment{Recolor $V_{target}$ using palette $\cP_{target}$} 
    \label{step: alg 3 - recoloring}
}  

$\clr \leftarrow$ \Call{\Split}{$\clr, \Flarge$} 
\Comment{Final Split}
\label{step: alg 3 - final split}
\\
\Return $\clr$
\end{algorithm}

\subsection*{Correctness proof and analysis}

Recall that $\sigma = \frac{n}{\Delta+1}$. We classify colors at any stage of the execution into the three families $\cP_{small}$, $\cP_{good}$ and $\cP_{large}$, defined in Eq. {\eqref{eq: color classes}}.
The next lemma follows directly from Lemma {\ref{lemma:safe_splitting_general}} with $\Flarge = \sigma$.

\begin{lemma}
\label{lemma:safe_splitting_second_algorithm}
The invocations of Procedure $\Split$ (with frequency threshold $\Flarge=\sigma$) to large colors $i \in\cP_{large}$ (where $\freq(i) \ge \sigma$) 
result in good colors (with frequency in $[\sigma/2, \sigma]$), creating no new small colors.
\end{lemma}


For each iteration $t$ of 
the loop, let $\cP_{small}^t$ be the set of small colors at the beginning of 
Step \ref{step: alg 3 - classification-start},
and let $\cP_{remove}^t$ be the set 
$\cP \setminus \cP_{target}$ defined in 
Steps \ref{step: alg 3 - classification-start} - \ref{step: alg 3 - classification-end}.

\begin{lemma}
\label{lemma:decay_s}
At least half of the colors in $\cP_{small}^t$ are in $\cP_{remove}^t$ and are eliminated in 
Step \ref{step: alg 3 - recoloring}.
Formally, $|\cP_{small}^t \cap \cP_{remove}^t| \ge \frac{1}{2}|\cP_{small}^t|$.
\end{lemma}

\begin{proof}
Let $q_t = |\cP_{small}^t \cap \cP_{target}^t|$ be the number of small colors retained in the target palette in iteration $t$.
Let $g = |\cP_{good}^t|$ and $s=|\cP_{small}^t|$ be the number of good and small colors after 
the $\Split$ step of iteration $t$ (recall that $\cP_{large}^t$ is empty after 
the split).

Summing the total frequencies, and recalling that $\freq(i) \le \sigma$ for $i \in \cP_{good}^t$ and $\freq(i) < \sigma/2$ for $i \in \cP_{small}^t$, we get
$$
g\cdot\sigma + s\cdot\sigma/2  \ge \tfreq(\cP_{good}^t)+\tfreq(\cP_{small}^t)=n,
$$
hence
\begin{equation}
\label{eq: g+s/2}
g + s/2 \ge n/\sigma = \Delta+1.
\end{equation}
Good colors have no lower frequency than small colors. Therefore, if $q_t>0$, then some small color is in $\cP_{target}^t$, and all good colors must also belong to $\cP_{target}^t$. In this case, $q_t=|\cP_{target}^t|-g=\Delta+1-g$.

If $q_t=0$ then all small colors belong to $\cP_{remove}^t$, and we are done. Otherwise, by Eq. {\eqref{eq: g+s/2}}, $q_t = \Delta+1 - g \le s/2$, so $|\cP_{small}^t\cap\cP_{remove}^t| = s - q_t \ge s/2$. The lemma follows.
\end{proof}
%
%

For the sake of the time analysis of Section~\ref{sec:congest-model} we note the following:

\begin{fact}
\label{lem:ildrc-runtime-interface}
Algorithm $\ILDRC$ makes one initial call to $\DRC$, then performs $O(\lg\Delta)$ iterations, each containing one call to $\Split$, one $\CA$ step for sorting and selecting $\cP_{target}$, and one call to $\Recolor$. It ends with one additional call to $\Split$.
\end{fact}



\begin{theorem}
\label{thm: ILDRC}
Procedure $\ILDRC(G)$ (Algorithm~\ref{alg:iterative_balanced_final}) produces 
a valid coloring with $\numclr\in [(\Delta+1), 2(\Delta+1)]$ colors 
and frequency range $[\sigma/2,\sigma]$. Its running times are $\Tseq(\ILDRC)=O(m\lg\Delta)$ and $\Tcc(\ILDRC)=O(\lg\Delta)$.
In the $\CONGEST$ model, the coloring is obtained in time  $\Tcongest(\ILDRC)=O(\lg\Delta\cdot(D+\Delta+\lg^5\lg n))$ with high probability.
\end{theorem}

\begin{proof}
We first prove that the loop runs for $T = O(\lg\Delta)$ iterations. Let $s_t$ be the number of Small colors at iteration $t$.
By Lemma \ref{lemma:safe_splitting_second_algorithm}, 
the split step never creates small colors.
By Lemma \ref{lemma:decay_s}, 
Step \ref{step: alg 3 - recoloring}
eliminates at least half of the existing small colors.
Thus, $s_{t+1} \le s_t / 2$. At the beginning of each iteration, the palette before the split has size $\Delta+1$, and the split step creates no new small colors; hence $s_t\le \Delta+1$ whenever small colors remain. Therefore after $O(\lg\Delta)$ iterations, $\cP_{small} = \emptyset$.
After the final split 
(Step \ref{step: alg 3 - final split}),
all colors satisfy $\freq(i)\in[\sigma/2,\sigma]$.
The number of colors $\numclr$ is bounded by
$\numclr \le
\frac{n}{\Fsmall}
	\le \frac{n}{\sigma/2} = 2(\Delta+1)$,
	%
and $\numclr \ge \frac{n}{\Flarge}\ge \frac{n}{\sigma}=\Delta+1$.
The runtime statement follows by Fact~\ref{lem:ildrc-runtime-interface} (see Section~\ref{sec:congest-model}).
\end{proof}

\section{Near-Equitable Coloring with Palette - Frequency Range Trade-off}
\label{sec:generalized-coloring}
 
In this section, we present a generalization that allows for a trade-off between the total number of colors and the maximum color frequency. By adjusting a parameter $k \ge 1$, we achieve a coloring with at most $(1 + 1/k)(\Delta+1)$ colors, where every color $i$ has frequency $\freq(i) \in [\sigma/2, 2k\sigma]$ in a total $\CC$ run-time of $O(\lg\Delta)$ rounds. 

This approach leverages the balanced state produced by Algorithm~\ref{alg:iterative_balanced_final} ($\ILDRC$) to aggressively reduce the palette size to $\Delta+1$ before performing a final relaxed split.

\subsection{Algorithm Description}

The algorithm, named $\PFTradeoff(G,k)$, first executes Algorithm $\ILDRC$ (Algorithm~ \ref{alg:iterative_balanced_final}) to obtain a high-quality equitable coloring. It then forces a reduction to exactly $\Delta+1$ color classes (potentially creating large classes) and finally splits any class exceeding the relaxed threshold $\lceil 2k\sigma\rceil$.

\begin{algorithm}[H]
\caption{$\PFTradeoff(G, k)$}
\label{alg:tradeoff}
\textbf{Input:} Graph $G=(V,E)$, parameter $k \ge 1$.\\
\textbf{Output:} Valid coloring $\clr$ with frequency range $[\sigma/2, 2k\sigma]$.\\
$\clr \leftarrow$ $\ILDRC(G)$ \Comment{Returns $\le 2(\Delta+1)$ classes, sizes in $[\sigma/2, \sigma]$}
\label{step:alg ILDRC}\\
Sort colors in $\cP$ by nonincreasing order of frequency: $\freq(i_1) \ge \freq(i_2) \ge \dots \ge \freq(i_m)$. \\
$\cP_{target} \leftarrow \{i_1, \dots, i_{\Delta+1}\}$ \Comment{Select the $\Delta+1$ most frequent colors}\\
$V_{target} \leftarrow \bigcup_{i \in \cP \setminus \cP_{target}} V(i)$\\
$\Recolor(V_{target}, \cP_{target})$ \Comment{Dissolve small classes into the top set}\\
$\cP \leftarrow \cP_{target}$ \Comment{Palette size is now exactly $\Delta+1$}\\
$\Flarge \leftarrow\lceil 2k\sigma\rceil$
\Comment{Final Relaxed Split}
\label{step:setting Flarge}
\\
$\clr \leftarrow$ $\Split(\clr, \Flarge)$\\
\label{step:enforce upper threshold}
\Return $\clr$
\end{algorithm}

\subsection{Correctness proof and analysis}

We now verify the bounds for the frequency range and the palette size.

\begin{lemma}
\label{lem:tradeoff-sizes}
For 
$k \ge 1$, the algorithm produces a coloring with
frequency range $[\lfloor\sigma/2\rfloor,\lceil 2k\sigma\rceil]$. 
\end{lemma}

\begin{proof}
Step \ref{step:setting Flarge} of the algorithm sets $\Flarge =\lceil 2k\sigma\rceil$, so by Lemma \ref{lemma:safe_splitting_general}, the upper frequency threshold is enforced explicitly by 
Step \ref{step:enforce upper threshold} of the algorithm.

Turning to the lower bound, we analyze the frequencies before and after the split.
After Step \ref{step:alg ILDRC}
(Algorithm~\ref{alg:iterative_balanced_final} ($\ILDRC$)), all colors satisfy $\freq(i) \ge \lfloor\sigma/2\rfloor$. 
In the recoloring stage,
the colors in $\cP_{target}$ only absorb vertices from $V_{target}$, so their frequencies strictly increase. Thus, before the final split (Step \ref{step:enforce upper threshold}), all colors in $\cP$ satisfy $\freq(i) \ge \lfloor\sigma/2\rfloor$.
In Step \ref{step:enforce upper threshold},
a color $i$ is split only if $\freq(i) > \lceil 2k\sigma\rceil$. By Lemma \ref{lemma:safe_splitting_general}, splitting a color with upper threshold $\Flarge=\lceil 2k\sigma\rceil$ produces new colors with frequency 
\begin{equation}
\label{eq:bound Flarge}
\freq(i_{new}) \ge \left\lceil \frac{\Flarge}{2} \right\rceil
= \left\lceil \frac{\lceil 2k\sigma\rceil}{2} \right\rceil
\ge k\sigma \ge \lfloor\sigma/2\rfloor,
\end{equation}
where the penultimate inequality follows from $\lceil 2k\sigma\rceil \ge 2k\sigma$ and the last one follows since $k \ge 1$. 
Thus, both the unsplit colors and the newly split colors respect the lower frequency threshold $\lfloor\sigma/2\rfloor$.
\end{proof}
\begin{lemma}
\label{lem:tradeoff-palette}
The number of color classes in the final coloring is at most $(1 + 1/k)(\Delta+1)$.
\end{lemma}

\begin{proof}
Prior to the final split,
the algorithm compacts the coloring into exactly $\Delta+1$ colors. The final phase splits only the colors in the set $S=\{i \mid \freq(i) > \lceil 2k\sigma\rceil\}$.
When a color $i\in S$ is split, let $p=p_i$ be the number of final colors descended from $i$. It introduces $p-1$ additional colors. By Eq. \eqref{eq:bound Flarge} (or Lemma \ref{lemma:safe_splitting_general}) and the integrality of $\Flarge=\lceil 2k\sigma\rceil$, each final descendant has frequency at least $k\sigma$. Therefore, $\freq(i) \ge p \cdot k\sigma$, which implies that the number of new colors contributed by this split is $p-1 < p \le \frac{\freq(i)}{k\sigma}$.
Summing over all split colors, the total number of new colors, $\numclr_{new}$, is bounded by the total frequency $n$ divided by the contribution of each split color, i.e.,
$\numclr_{new} < \sum_{i \in S} \frac{\freq(i)}{k\sigma} \le \frac{n}{k\sigma}~.$
Substituting $\sigma = \frac{n}{\Delta+1}$, we get
$\numclr_{new} < 
\frac{\Delta+1}{k}$.
Thus, the total number of colors is bounded by $(\Delta+1) + \frac{\Delta+1}{k} = (1 + 1/k)(\Delta+1)$.
\end{proof}

This theorem establishes a 
trade-off between the palette size and the maximum color frequency. While the iterative balanced coloring achieves a tight frequency balance of $[\sigma/2, \sigma]$ at the cost of using up to $2(\Delta+1)$ colors, the generalized algorithm allows us to push the palette size arbitrarily close to the ideal lower bound of $\Delta+1$ by increasing $k$. Crucially, this reduction in palette size does not compromise the lower threshold on frequencies; all color frequencies remain bounded below by $\sigma/2$. 


For the sake of the time analysis of Section~\ref{sec:congest-model} we note the following:

\begin{fact}
\label{lem:pftradeoff-runtime-interface}
After the initialization by $\ILDRC$, Algorithm $\PFTradeoff$ performs one $\CA$ step to sort colors and choose $\cP_{target}$, one call to $\Recolor$, and one final call to $\Split$.
\end{fact}



The next theorem now follows from Lemmas~\ref{lem:tradeoff-sizes}, \ref{lem:tradeoff-palette} and Fact~\ref{lem:pftradeoff-runtime-interface}. 


\begin{theorem}
\label{thm:tradeoff-main}
For a graph $G$ and an integer parameter $k\ge 1$, $\PFTradeoff$ (Algorithm~\ref{alg:tradeoff}) computes a valid coloring $\clr$ of $G$ with
palette size $\numclr \le \left(1 + \frac{1}{k}\right)(\Delta+1)$ and
frequency range $\left[\frac{\sigma}{2}, 2k\sigma\right]$.
Its running times are $\Tseq(\PFTradeoff)=O(m\lg\Delta)$ and $\Tcc(\PFTradeoff)=O(\lg\Delta)$.
In the $\CONGEST$ model, the coloring is obtained in time  $\Tcongest(\PFTradeoff)=O(\lg\Delta\cdot(D+\Delta+\lg^5\lg n))$ with high probability.
\end{theorem}

\subsection{The SMALL Initialization Variant}
\label{subsec:pftradeoff-small}

A second instantiation of the same split \& Recolor block
uses the faster low-threshold initialization $\NBC$ (Algorithm~\ref{alg:balanced_coloring}) from Section~\ref{s: longn time 3Delta palette [sigma/3, 2sigma/3]}.
Define $\PFTradeoffSmall(G,k)$ to be the algorithm obtained from Algorithm~\ref{alg:tradeoff} ($\PFTradeoff$) by replacing line~\ref{step:alg ILDRC} with the call $\clr\gets\NBC(G,1/3)$. 
All subsequent lines are unchanged: the algorithm keeps the $\Delta+1$ most frequent colors, recolors the deleted color classes into this target palette, and finally invokes $\Split$ with threshold $\Flarge=\lceil2k\sigma\rceil$.

For the sake of the time analysis of Section~\ref{sec:congest-model} we note the following:

\begin{fact}
\label{lem:pftradeoff-small-runtime-interface}
Algorithm $\PFTradeoffSmall$ has the same accounting as Fact~\ref{lem:pftradeoff-runtime-interface}, with the initialization $\ILDRC$ replaced by $\NBC$.
\end{fact}




The proof of the following theorem goes along similar lines to that of Lemmas~\ref{lem:tradeoff-sizes}, \ref{lem:tradeoff-palette} and Theorem \ref{thm:tradeoff-main}, and is therefore omitted.

\begin{theorem}
\label{thm:tradeoff-small-main}
For a graph $G$ and an integer parameter $k\ge1$, $\PFTradeoffSmall$ computes a valid coloring $\clr$ of $G$ with palette size $\numclr\le\left(1+\frac1k\right)(\Delta+1)$ and frequency range $[\sigma/3,2k\sigma]$.
Its running times are $\Tseq(\PFTradeoffSmall)=O(m)$ and $\Tcc(\PFTradeoffSmall)=O(1)$.
In the $\CONGEST$ model, the coloring is obtained in time  $\Tcongest(\PFTradeoffSmall)=O(D+\Delta+\lg^5\lg n)$ with high probability.
\end{theorem}


\section{Near-equitable coloring with Palette - Time Trade-off}
\label{sec:fast-converging}

In this section, we propose an algorithm that achieves a coloring with palette size at most $(1+1/k)(\Delta+1)$ and frequency range $[\sigma/2, 2\sigma]$.
This algorithm provides a middle ground between the tight palette of the exact $(\Delta+1)$-coloring 
and the speed of the earlier approximations. Specifically, by relaxing the palette constraint slightly by a factor of $1/k$, we ensure the algorithm terminates in $O(k)$ iterations, with total $\CC$ time complexity $O(k+\lg\Delta)$ rounds. 

\subsection{Algorithm Description}

Algorithm $\PTTradeoff$ iteratively compacts the coloring into the largest $\Delta+1$ classes and then handles overflows. It continues this process until the number of color classes falls within the target range.

\begin{algorithm}[H]
\caption{$\PTTradeoff(G, k)$}
\label{alg:fast-converge}
\textbf{Input:} Graph $G=(V,E)$, parameter $k \ge 1$.\\
\textbf{Output:} Valid coloring $\clr$ with $\le (1+1/k)(\Delta+1)$ colors, sizes in $[\sigma/2, 2\sigma]$.
\\
$\clr \leftarrow$ $\ILDRC(G)$ \Comment{Initial coloring: $\le 2(\Delta+1)$ colors}
\label{step:fast-converge-init}\\
$\Flarge \leftarrow 2\sigma$
\\
\While{$\numclr > (1+1/k)(\Delta+1)$}
{
    Sort colors in $\cP$ by nonincreasing order of frequency:
   $\freq(i_1) \ge \freq(i_2) \ge \dots \ge \freq(i_m)$. \\
    $\cP_{target} \leftarrow \{i_1, \dots, i_{\Delta+1}\}$\\
    $V_{target} \leftarrow \bigcup_{i \in \cP \setminus \cP_{target}} V(i)$\\
    $\Recolor(V_{target}, \cP_{target})$\\
$\clr \leftarrow$ $\Split(\clr, \Flarge)$ \\
}  
\Return $\clr$
\end{algorithm}

\subsection{Correctness proof and analysis}

We first verify that the color frequencies remain within the desired range following each execution of Procedure $\Split$ throughout the run. 

\begin{lemma}
\label{lem:freq range P T Tradeoff}
In any execution of $\PTTradeoff$ (Algorithm \ref{alg:fast-converge}),
following each invocation of Procedure $\Split$,
every color $i$ satisfies $\freq(i)\in[\sigma/2,2\sigma]$.
\end{lemma}

\begin{proof}
The initial coloring ($\ILDRC$, Algorithm~\ref{alg:iterative_balanced_final}) satisfies the bounds $[\sigma/2, \sigma] \subset [\sigma/2, 2\sigma]$.
In 
the recoloring step,
vertices are added to existing colors, so their frequencies only increase; thus, the lower bound $\sigma/2$ is preserved.
In 
the splitting step,
any color with $\freq(i) > 2\sigma$ is split. By Lemma \ref{lemma:safe_splitting_general}, splitting a color with $\freq(i) > 2\sigma$ results in new colors with frequency at least $\lfloor \freq(i)/2 \rfloor \ge \sigma$. Since $\sigma > \sigma/2$, the lower bound is preserved.
Finally, the splitting process explicitly enforces the upper bound $2\sigma$ at the end of every iteration.
\end{proof}

We now analyze the convergence rate, proving that the algorithm terminates quickly.

\begin{lemma}
\label{thm:fast-converge-runtime}
Algorithm $\PTTradeoff$
terminates in $O(k)$ iterations. 
\end{lemma}

\begin{proof}
To prove the iteration bound, we categorize colors based on their frequency relative to $\sigma$. For this analysis only,
define a color $i$ as
\emph{large} if $\freq(i) \ge \sigma$ and
\emph{small} 
otherwise.

The total number of vertices in the graph is $n = (\Delta+1) \sigma$. Since every large color has $\freq(i) \ge \sigma$, the 
number of large colors at any point in the execution is bounded by \begin{equation}
\label{eq:Nlarge}
N^{large} \le n/\sigma = \Delta+1.
\end{equation}


We now track the number of small colors, denoted $N^{small}_t$, at the beginning of iteration $t$. We analyze the two steps of the loop.
In each iteration, the algorithm retains the $\Delta+1$ most frequent colors and dissolves the rest. The loop condition ensures the current total number of colors is $\numclr_t > (1 + 1/k)(\Delta+1)$.
The number of dissolved colors is:
$$ N^{dissolved} = \numclr_t - (\Delta+1) > \frac{1}{k}(\Delta+1).$$
By Eq. \eqref{eq:Nlarge}, the $\Delta+1$ highest frequency colors \emph{must} include all existing large colors. Consequently, \textbf{all} $N^{dissolved}$ colors that are dropped must be small colors.

The algorithm splits colors with $\freq(i) > 2\sigma$. By Lemma \ref{lemma:safe_splitting_general}, splitting such a color
creates new colors with frequency at least $\sigma$. Thus, the splitting process 
never creates a small color.

Combining these observations, the number of small colors strictly decreases in each iteration. Specifically, at least $(\Delta+1)/k$ small colors are removed in 
the recoloring stage, and no small colors are added in the splitting step. Hence,
$s_{t+1} \le s_t - (\Delta+1)/k$. The initial coloring by Algorithm $\ILDRC(G)$ uses at most $2(\Delta+1)$ colors, so $s_0 \le 2(\Delta+1)$.
It follows that
the maximum number of iterations satisfies $\frac{s_0}{(\Delta+1)/k}  
=O(k)$. 
%
\end{proof}

For the sake of the time analysis of Section~\ref{sec:congest-model} we note the following:

\begin{fact}
\label{lem:pttradeoff-runtime-interface}
Algorithm $\PTTradeoff$ first invokes $\ILDRC$. By Lemma~\ref{thm:fast-converge-runtime}, it then performs $O(k)$ loop iterations, each containing one $\CA$ step for sorting and selecting the retained palette, one call to $\Recolor$, and one call to $\Split$.
\end{fact}



\begin{theorem}
\label{thm:summary-fast-converge}
For every $k\ge1$, Algorithm~\ref{alg:fast-converge} ($\PTTradeoff$) returns a valid coloring with at most $(1+1/k)(\Delta+1)$ colors and frequency range $[\sigma/2,2\sigma]$. Its running times are $\Tseq(\PTTradeoff)=O(m(k+\lg\Delta))$ and $\Tcc(\PTTradeoff)=O(k+\lg\Delta)$.
In the $\CONGEST$ model, the coloring is obtained in time  $\Tcongest(\PTTradeoff)=O((\lg\Delta+k)\cdot(D+\Delta+\lg^5\lg n))$ with high probability.
\end{theorem}

\begin{proof}
The bound on the palette size follows from the stopping condition of the loop, and the frequency range bounds follow from Lemma \ref{lem:freq range P T Tradeoff}. 
The runtime statement follows from Fact~\ref{lem:pttradeoff-runtime-interface} and is calculated in Section~\ref{sec:congest-model}.
\end{proof}

\subsection{The SMALL Initialization Variant}
\label{subsec:pttradeoff-small}

A replacement similar to that of Section {\ref{subsec:pftradeoff-small}} gives a faster low-threshold version of the palette-time tradeoff algorithm.
Define $\PTTradeoffSmall(G,k)$ to be the algorithm obtained from Algorithm~\ref{alg:fast-converge} ($\PTTradeoff$) by replacing line~\ref{step:fast-converge-init} with the call $\clr\gets\NBC(G,1/3)$. 
The stopping threshold $(1+1/k)(\Delta+1)$, 
the repeated recoloring to the $\Delta+1$ largest colors, and the split threshold $\Flarge=2\sigma$ are left unchanged.

For the sake of the time analysis of Section~\ref{sec:congest-model} we note the following:

\begin{fact}
\label{lem:pttradeoff-small-runtime-interface}
Algorithm $\PTTradeoffSmall$ first invokes $\NBC$. Since the $\NBC(G,1/3)$ initialization starts with at most $3(\Delta+1)$ colors, and every loop iteration dissolves at least $(\Delta+1)/k$ colors before the split, the algorithm performs $O(k)$ loop iterations. Each iteration contains one $\CA$ step, one call to $\Recolor$, and one call to $\Split$.
\end{fact}



The proof of the following theorem goes along similar lines to that of Lemmas \ref{lem:freq range P T Tradeoff}, \ref{thm:fast-converge-runtime} and Theorem \ref{thm:summary-fast-converge}, and is therefore omitted.

\begin{theorem}
\label{thm:summary-fast-converge-small}
For every $k\ge1$, $\PTTradeoffSmall$ returns a valid coloring with at most $(1+1/k)(\Delta+1)$ colors and frequency range $[\sigma/3,2\sigma]$.
Its running times are $\Tseq(\PTTradeoffSmall)=O(mk)$ and $\Tcc(\PTTradeoffSmall)=O(k)$.
In the $\CONGEST$ model, the coloring is obtained in time  $\Tcongest(\PTTradeoffSmall)=O(k\cdot(D+\Delta+\lg^5\lg n))$ with high probability.
\end{theorem}

\section{Capped palette reduction}
\label{sec:hard-capped-open-reduction}
\label{sec:improved_coloring}

This section presents Algorithm $\OpenRecolor$, which refines our previous approach to the palette-reduction step. Starting from the basic $\ILDRC$
coloring of Section \ref{s: ILDRC}, it reduces the palette from at most $2(\Delta+1)$ colors to
$\beta(\Delta+1)$, where $1<\beta<2$, while preserving frequencies throughout the process in the range $[\sigma/2, \Flarge]$ for the integer cap $\Flarge = \lceil\alpha\sigma\rceil$.

%
%

\subsection{Overview and goal}


Let us start with describing the main idea for the improved solution. Algorithm $\ILDRC$
(Section \ref{s: ILDRC})
was based on alternating between \emph{recoloring phases}, which reduced the number of colors at the cost of increasing the frequencies of the remaining colors, and \emph{splitting phases}, which increased the number of colors. The idea at the basis of the new approach is to prevent the need for splitting, thus reducing the number of phases. To achieve that, we no longer treat all colors with frequency $\le \Flarge$ as valid targets in the recoloring stage. Rather, we forbid the use of a color $c$ if its frequency $\freq(c)$ is so close to $\Flarge$ that there's a risk that it might exceed $\Flarge$ if $c$ is used in the recoloring stage. 

Concretely, we
fix a threshold $\Fcap=\varphi \cdot \Flarge$ for a suitably chosen constant $\frac12<\varphi<1$, such that colors whose frequency falls in the range $[\Fcap,\Flarge]$ are considered \emph{closed} and cannot be used for recoloring.
An uncolored vertex $v$ cannot choose a color 
already used by a colored neighbor, i.e., a color from 
\begin{equation*}
\cP_{neig}(v) = \{\clr(u) \mid u\in N(v)\setminus V_{target}\},
\end{equation*}
so its \emph{nonconflicting palette} is 
\begin{equation}
\label{eq: def Ptarget}
\cP_{target}(v) = \cP_{target} \setminus \cP_{neig}(v)~.
\end{equation} 
Moreover, it must obey the additional capping restriction, namely, it is only allowed to select an \emph{open} color $c\in\cP_{target}(v)$, namely, such that $\freq(c)<\Fcap$.

This condition by itself does not guarantee that the resulting coloring will obey the upper frequency threshold $\Flarge$, as it might happen that many pairwise non-adjacent uncolored vertices select the same open color $c$, causing $\freq(c)$ to exceed $\Flarge$.

To overcome this difficulty, the algorithm imposes the threshold $\Flarge$ \emph{deterministically}; it counts the number of vertices that selected the color $c$, and allows only $\Flarge-\freq(c)$ uncolored vertices to use it. The remaining vertices remain uncolored and will try again in subsequent recoloring iterations.

A natural concern is that this might run the risk of leaving some vertices uncolored forever. As shown in our analysis, a careful choice of the threshold $\Fcap$ ensures that this risk never materializes, and every uncolored vertex gets recolored eventually.

\subsection{Parameters and algorithm}

Let $\beta=1+\ell/(\Delta+1)$, where $1\le\ell\le\Delta$, and use the
integer cap $\Flarge=\lceil\alpha\sigma\rceil$.
The target palette has size
$\numclr_{target}=\beta(\Delta+1)=\Delta+1+\ell$ colors.

Algorithm $\OpenRecolor$ deletes the least frequent colors and recolors their vertices into the retained target palette.  For an \emph{open-capacity} parameter $\varphi\in(0,1)$, set
$\Fcap=\lfloor\varphi\alpha\sigma\rfloor$ and use the capacity gap
$\Flarge-\Fcap$. Vertices propose only open colors in $\cP_{target}(v)$. Once candidates for color $c$ are listed, the algorithm \emph{accepts} at most
$\Flarge-\freq(c)$ of the candidates (namely, allows them to use the color) and rejects the others.

The parameters used in this
section are required to satisfy the following conditions:
\begin{eqnarray}
& \beta=1+\frac{\ell}{\Delta+1}
\quad\mbox{for an integer }1\le\ell\le\Delta,
\quad\mbox{hence }1<\beta<2,
\label{eq:hc-beta}
\\
& \alpha\ge1,
\label{eq:hc-alpha}
\\
& \frac12<\varphi<1,
\label{eq:hc-phi}
\\
& 2\varphi\alpha\cdot\frac{\ell}{\Delta+1}>1 .
\label{eq:hc-open-slack}
\end{eqnarray}
Since $\Flarge$ is a positive integer,
$\varphi<1$ and $\Fcap$ is rounded down, 
$\Flarge-\Fcap \ge 1$.

The activation probability $p_0$ is chosen from the same parameters and defined in
Algorithm~\ref{alg:open-recolor} ($\OpenRecolor$).

Algorithm~\ref{alg:open-recolor} ($\OpenRecolor$) gives the formal code. 

\begin{center}
\refstepcounter{algocf}
\label{alg:open-recolor}
\label{alg:hard-cap-wrapper}
\label{alg:open-cap-recolor}
\begingroup
\footnotesize
\newcount\OpenRecolorLineNo
\OpenRecolorLineNo=0
\newcommand{\AlgoLine}[2]{%
  \global\advance\OpenRecolorLineNo by 1%
  \par\noindent
  \makebox[1.8em][r]{\scriptsize\the\OpenRecolorLineNo}\hspace{0.45em}%
  \parbox[t]{\dimexpr\linewidth-2.25em\relax}{\hangindent=#1\hangafter=1\noindent\hspace*{#1}#2}%
  \par}
\begin{minipage}{\linewidth}
\hrule
\vspace{2pt}
\noindent{\bfseries Algorithm~\thealgocf:} $\OpenRecolor(G,\alpha,\beta,\varphi)$
\vspace{2pt}
\hrule
\vspace{2pt}
\begin{tabular}{@{}p{0.462\linewidth}@{\hspace{0.03\linewidth}}p{0.462\linewidth}@{}}
\begin{minipage}[t]{\linewidth}
\raggedright
\AlgoLine{0pt}{\textbf{Input:} graph $G=(V,E)$, \\ parameters $\alpha,\beta,\varphi$ 
where $\beta=1+\ell/(\Delta+1)$.}
\AlgoLine{0pt}{\textbf{Output:} a valid coloring with at most $\numclr_{target}=\Delta+1+\ell$ colors, and frequencies $\in [\sigma/2, \lceil\alpha\sigma\rceil]$.}
\AlgoLine{0pt}{\textbf{Setup and palette reduction.}}
\AlgoLine{0pt}{$\clr\gets\ILDRC(G)$ \quad{\scriptsize (frequencies in $[\sigma/2,\sigma]$)}}
\AlgoLine{0pt}{$\Flarge\gets\lceil\alpha\sigma\rceil$, \quad $\Fcap\gets\lfloor\varphi\alpha\sigma\rfloor$} 
\AlgoLine{0pt}{$\delta\gets\beta-1-\frac{2-\beta}{2\varphi\alpha-1}$}
\AlgoLine{0pt}{$p_0\gets\frac12\min\{1,\delta,\delta(1-\varphi)\alpha/(1-\beta/2)\}$}
\AlgoLine{0pt}{$\numclr_{target}\gets\beta(\Delta+1)=\Delta+1+\ell$}
\AlgoLine{0pt}{\textbf{if} $\numclr\le\numclr_{target}$ \textbf{then}}
\AlgoLine{1.2em}{\textbf{return} $\clr$}
\AlgoLine{0pt}{Sort colors so that $\freq(i_1)\le\cdots\le\freq(i_{\numclr})$}
\AlgoLine{0pt}{$\numclr_{remove}\gets\numclr-\numclr_{target}$}
\AlgoLine{0pt}{$\cP_{remove}\gets\{i_1,\ldots,i_{\numclr_{remove}}\}$}
\AlgoLine{0pt}{$V_{target}\gets\bigcup_{i\in\cP_{remove}}V(i)$}
\AlgoLine{0pt}{$\cP_{target}\gets\cP\setminus\cP_{remove}$}
\AlgoLine{0pt}{Unset the colors of all vertices in $V_{target}$}
\end{minipage}
&
\begin{minipage}[t]{\linewidth}
\raggedright
\AlgoLine{0pt}{\textbf{Open recoloring phase.}}
\AlgoLine{0pt}{\textbf{while} $V_{target}\neq\emptyset$ \textbf{do}}
\AlgoLine{1.2em}{$\cP_{open}\gets\{c\in\cP_{target}\mid \freq(c)\le \Fcap\}$}
\AlgoLine{1.2em}{\textbf{for each} $v\in V_{target}$ in parallel \textbf{do}}
\AlgoLine{2.4em}{$\cP_{neig}(v)\gets\{\clr(u)\mid u\in N(v)\setminus V_{target}\}$}
\AlgoLine{2.4em}{$\cP_{target}(v)\gets \cP_{target}\setminus\cP_{neig}(v)$}
\AlgoLine{2.4em}{$\cP_{open}(v)\gets\cP_{target}(v)\cap\cP_{open}$}
\AlgoLine{2.4em}{Become active with probability $p_0$}
\AlgoLine{2.4em}{\textbf{if} $v$ is active and $\cP_{open}(v)\neq\emptyset$ \textbf{then}}
\AlgoLine{3.6em}{Pick $c_v\in\cP_{open}(v)$ uniformly at random and broadcast $\textsf{Propose}(c_v)$}
\AlgoLine{3.6em}{If no neighbor in $V_{target}$ proposes $c_v$, then $v$ becomes a candidate for color $c_v$}
\AlgoLine{1.2em}{\textbf{for each color} $c\in\cP_{target}$ in parallel \textbf{do}}
\AlgoLine{2.4em}{Let $C_c$ be the candidates for $c$}
\AlgoLine{2.4em}{
Draw a fresh random ranking for the vertices of $C_c$.}
\AlgoLine{2.4em}{$C_c^{add}\gets$ the set of $\min\{|C_c|,\Flarge-\freq(c)\}$ highest-ranked candidates}
\AlgoLine{2.4em}{For every $v\in C_c^{add}$, set $\clr(v)\leftarrow c$}
\AlgoLine{1.2em}{$V_{target}\gets V_{target}\setminus\bigcup_{c\in\cP_{target}} C_c^{add}$}
\AlgoLine{1.2em}{Update all frequencies}
\AlgoLine{0pt}{\textbf{return} $\clr$}
\end{minipage}
\end{tabular}
\vspace{2pt}
\hrule
\end{minipage}
\endgroup
\end{center}


\subsection{Correctness proof and analysis}

\begin{observation}
\label{obs:p_0}
For $\delta$ and $p_0$ as defined in Algorithm~\ref{alg:open-recolor} ($\OpenRecolor$),
the quantities $\Flarge,\Fcap,\delta,p_0$ are well-defined,
$\Flarge-\Fcap\ge1$, $\delta>0$, and $0<p_0\le1/2$.
\end{observation}

\begin{proof}
By Conditions~\eqref{eq:hc-alpha} and~\eqref{eq:hc-phi}, the denominator
$2\varphi\alpha-1$ in the definition of $\delta$ is positive.  Moreover,
since $\beta-1=\ell/(\Delta+1)$,
\begin{equation}
\label{eq: delta>0}
\delta ~=~ \beta-1-\frac{2-\beta}{2\varphi\alpha-1} ~=~ \frac{2\varphi\alpha(\beta-1)-1}{2\varphi\alpha-1} ~>~ 0,
\end{equation}
where the final inequality follows since Condition~\eqref{eq:hc-open-slack} gives
$2\varphi\alpha(\beta-1)>1$. 
The denominator in the third term in the definition of $p_0$ in Algorithm~\ref{alg:open-recolor} ($\OpenRecolor$) is positive because $\beta<2$ by Eq.~\eqref{eq:hc-beta}.  Since $\Flarge$ is a positive integer and $\varphi<1$, we have
$\Fcap=\lfloor\varphi\alpha\sigma\rfloor\le\Flarge-1$,
and hence $\Flarge-\Fcap \ge 1$.  Consequently,
Conditions~\eqref{eq:hc-beta}--\eqref{eq:hc-open-slack} yield the claims.
\end{proof}

\begin{lemma}
\label{lem:open-cap-safety}
At every point in Algorithm~\ref{alg:open-recolor} ($\OpenRecolor$), the partial coloring is
valid and every target color $c\in\cP_{target}$ satisfies
$\freq(c)\le\Flarge$.
Consequently, if the procedure terminates, it returns a valid coloring whose
target-color frequencies are at most $\Flarge$.
\end{lemma}

\begin{proof}
A vertex proposes only a color not used by any already-colored neighbor.  Two adjacent uncolored vertices proposing the same color cannot both become candidates, because both detect a conflict.  Hence accepting candidates
preserves the validity of the coloring.

For the upper bound, fix a color $c$ and an iteration.  The algorithm accepts at
most $\Flarge-\freq(c)$ candidates for $c$, so after the round its frequency is
still at most $\Flarge$.  This proves the invariant by induction over the
iterations.
\end{proof}

\begin{lemma}
\label{lem:hard-cap-uncolored-volume}
After the palette-reduction step of Algorithm~\ref{alg:open-recolor} ($\OpenRecolor$),
$|V_{target}|\le (1-\beta/2)n$.
\end{lemma}

\begin{proof}
The target palette has size $\numclr_{target}=\beta(\Delta+1)$.  Each retained
target color has frequency at least $\sigma/2$, because the input coloring
from $\ILDRC$ has frequencies in $[\sigma/2,\sigma]$.  Thus the number of
vertices colored by the retained colors satisfies
\[
\freq(\cP_{target}) ~\ge~
\numclr_{target}\cdot\frac{\sigma}{2}
~=~
\beta(\Delta+1)\cdot\frac{\sigma}{2}
~=~ \frac{\beta n}{2}~.
\]
All other vertices are placed in $V_{target}$.
\end{proof}

\begin{lemma}
\label{lem:many-open-choices}
Assume Conditions~\eqref{eq:hc-beta}--\eqref{eq:hc-open-slack}, and let
$\delta$ be the value defined in Algorithm~\ref{alg:open-recolor} ($\OpenRecolor$).  In every round of
Algorithm~\ref{alg:open-recolor} ($\OpenRecolor$), every uncolored vertex $v$ satisfies
$|\cP_{open}(v)| \ge \delta(\Delta+1)+1$.
\end{lemma}

\begin{proof}
Fix an iteration, and let $\cP_{closed}=\cP_{target} \setminus \cP_{open}$. 
Every $c\in\cP_{closed}$ satisfies $\freq(c)>\varphi\alpha\sigma$ because $\Fcap=\lfloor\varphi\alpha\sigma\rfloor$
and frequencies are integral.
Every target color has
frequency at least $\sigma/2$, since target colors are never dissolved and
only gain vertices.  Therefore,
\[
n \ge \numclr_{closed} \cdot \varphi\alpha\sigma
    +(\numclr_{target}-\numclr_{closed})\cdot\frac{\sigma}{2}~.
\]
Substituting $\numclr_{target}=\beta(\Delta+1)$, dividing by $\sigma$ and rearranging gives
\begin{equation}
\label{eq: Xclosed ub alg open-recolor}
\numclr_{closed} ~\le~ \frac{1-\beta/2} {\varphi\alpha-1/2}(\Delta+1) ~=~ \frac{2-\beta} {2\varphi\alpha-1} (\Delta+1).
\end{equation}

The set $\cP_{neig}(v)$ of colors used by $v$'s colored neighbors is of size $\numclr_{neig}(v)\le \Delta$,
so by Eq. {\eqref{eq: def Ptarget}}, $\numclr_{target}(v)$ 
satisfies
\begin{equation}
\label{eq: Xtarget lb alg open-recolor}
\numclr_{target}(v)\ge \numclr_{target}-\Delta ~=~ (\beta-1)(\Delta+1)+1.
\end{equation}
By Eq. {\eqref{eq: Xclosed ub alg open-recolor}} and {\eqref{eq: Xtarget lb alg open-recolor}},
\[
\numclr_{open}(v) ~\ge~ \numclr_{target}(v)-\numclr_{closed} ~\ge~ \left(\beta-1-\frac{2-\beta}{2\varphi\alpha-1}\right)(\Delta+1)+1 ~=~ \delta(\Delta+1)+1
\]
by the definition of $\delta$.
\end{proof}

For a recoloring round $t$, let $\HIST_t$ denote the history of the
execution up to the beginning of round $t$.  Let
$E_{\mathrm{GoodHist}}(t)$ be the event that no bad event occurred in iterations
$1,\ldots,t-1$, where a bad event means that (a) the partial coloring became
invalid, (b) some target color exceeded the cap $\Flarge$, or (c) some
uncolored vertex had fewer than $\delta(\Delta+1)+1$ open locally valid colors at the beginning of an iteration.

We observe that
any history $\HIST_t$ reached by Algorithm~\ref{alg:open-recolor} ($\OpenRecolor$) by the end of round $t-1$ satisfies $E_{\mathrm{GoodHist}}(t)$ deterministically. This can be verified by induction on $t$. For $t=1$, note that before the first recoloring round, uncolored vertices cannot create a conflict among the vertices that remain colored, all target colors have frequency at most $\sigma\le\Flarge$, and Lemma~\ref{lem:many-open-choices} establishes condition (c) in the
definition of $E_{\mathrm{GoodHist}}(t)$.
For the inductive step, assuming that $\HIST_t$ satisfies
$E_{\mathrm{GoodHist}}(t)$, Lemma~\ref{lem:open-cap-safety} shows that round
$t$ preserves conditions (a) and (b), and Lemma~\ref{lem:many-open-choices}
establishes condition (c) for the resulting history $\HIST_{t+1}$.

In the next lemma, conditioning on a fixed history $\HIST_t$ (which satisfies
$E_{\mathrm{GoodHist}}(t)$ as shown above),
the partial coloring, the uncolored set, all frequencies, and all palettes are fixed; the only remaining randomness is the fresh random choices made in round $t$.  

\begin{lemma}
\label{lem:open-cap-progress}
Assume Conditions~\eqref{eq:hc-beta}--\eqref{eq:hc-open-slack}, and set
$p_0$ as in Algorithm~\ref{alg:open-recolor} ($\OpenRecolor$).
For every round $t$ and every history $\HIST_t$,
each vertex that is still uncolored at the beginning of round $t$ becomes colored in that round with conditional probability at least $p_0/4$.
\end{lemma}

\begin{proof}
Fix an iteration $t$, condition on a history $\HIST_t$ (which satisfies
$E_{\mathrm{GoodHist}}(t)$ as shown above), and fix an uncolored vertex $v$.
In this proof all probabilities and expectations are with respect to the random
choices made in round $t$, conditioned on this history.
Set $A=\delta(\Delta+1)+1$.  By Lemma~\ref{lem:many-open-choices},
$\numclr_{open}(v)\ge A$.

Let $E_{active}(v)$ be the event that $v$ is active.
For any uncolored vertex $x$ and color $c$, let
$E_{x\to c}$ be the event that $x$ is active and proposes $c$
in this round. This event already includes activation, so no additional
conditioning on $E_{active}(x)$ is used in the proposal bounds.
Fix a color $c\in\cP_{open}(v)$. For any uncolored neighbor $u$,
$\mathbb P[E_{u\to c}\mid E_{v\to c}]\le \frac{p_0}{A}$.
Indeed, if $c\notin\cP_{open}(u)$ then this probability is $0$,
and otherwise it is at most
$p_0/|\cP_{open}(u)|\le p_0/A$.  By the union bound
over at most $\Delta<\Delta+1$ neighbors, 
\begin{equation}
\label{eq:pv-conflict}
\mathbb P[\text{$v$ has a proposal conflict}\mid E_{v\to c}]
\le
\Delta\cdot\frac{p_0}{A}
>
\frac{p_0}{\delta}
\le
\frac12~,
\end{equation}
where the last inequality follows from
Eq. {\eqref{eq: delta>0}} and the definition of
$p_0$ in Algorithm~\ref{alg:open-recolor} ($\OpenRecolor$), which
gives $p_0\le\delta/2$.
Hence, conditioned on $E_{v\to c}$, vertex $v$ becomes a candidate with
probability at least $1/2$.

It remains to bound the probability that color $c$ accepts $v$.  Let $Y$
be the number of uncolored vertices outside $N(v)\cup\{v\}$ that propose
$c$ in this round.  By Lemma~\ref{lem:hard-cap-uncolored-volume}, the current uncolored set has size at most $(1-\beta/2)n$. For every such vertex $w$,
\[
\mathbb P[E_{w\to c}\mid E_{v\to c}]\le \frac{p_0}{A}~.
\]
Therefore
\[
\mathbb E[Y\mid E_{v\to c}]
\le
(1-\beta/2)n\cdot\frac{p_0}{A}
\le
\frac{(1-\beta/2)\sigma p_0}{\delta}~.
\]
By the definition of $p_0$ in Algorithm~\ref{alg:open-recolor} ($\OpenRecolor$),
\[
\mathbb E[Y\mid E_{v\to c}]
\le \frac12(1-\varphi)\alpha\sigma .
\]
Since
\[
\Flarge-\Fcap
=
\lceil\alpha\sigma\rceil-\lfloor\varphi\alpha\sigma\rfloor
\ge \alpha\sigma-\varphi\alpha\sigma
= (1-\varphi)\alpha\sigma,
\]
Markov's inequality gives
\begin{equation}
\label{eq:y-ge-cap-gap}
\mathbb P[Y\ge \Flarge-\Fcap
   \mid E_{v\to c}]
~\le~
\frac{\mathbb E[Y\mid E_{v\to c}]}{\Flarge-\Fcap}
~\le~ \frac12.
\end{equation}
The random choices of vertices outside $N(v)\cup\{v\}$ are independent of the
random choices of the uncolored neighbors of $v$, conditioned on
$E_{v\to c}$.  Because $c$ is open, its remaining
capacity is at least $\Flarge-\Fcap$.  If no neighbor of $v$ proposes $c$ and
$Y<\Flarge-\Fcap$, then even after including $v$, at most $\Flarge-\Fcap$ vertices propose
$c$.  Thus all candidates for $c$ can be accepted, and in particular $v$ is
accepted.  Therefore, by Eqs.~\eqref{eq:pv-conflict} and~\eqref{eq:y-ge-cap-gap},
conditioned on $E_{v\to c}$, the probability that $v$ is colored is at least
$1/4$.

This lower bound is uniform over all possible choices
$c\in\cP_{open}(v)$.  Hence, by the law of total probability over the random
color chosen by $v$, conditioned only on $v$ being active, the probability
that $v$ is colored is also at least $1/4$.  Multiplying by the activation
probability $p_0$, the one-round success probability is at least $p_0/4$. 
\end{proof}

The preceding lemmas establish the deterministic safety, the availability of
open colors, and the one-round progress bound.  We now combine them into the
parameterized guarantee for Algorithm~\ref{alg:open-recolor}
($\OpenRecolor$), and then derive the asymptotic palette--frequency tradeoff as
two corollaries.

For the sake of the time analysis of Section~\ref{sec:congest-model} we note the following:

\begin{fact}
\label{lem:open-recolor-runtime-interface}
Algorithm $\OpenRecolor$ first invokes $\ILDRC$. By Lemma~\ref{lem:open-cap-progress}, with high probability it then performs $O(\lg n/p_0)$ open-capacity iterations. Each such iteration updates open colors and frequencies using $\CA$ and admits candidates according to the residual capacities using $\QA$. The local proposal and conflict-test operations are handled within these iterations and do not introduce an additional primitive in the accounting of Section~\ref{sec:congest-model}. When $\alpha,\beta,\varphi$ are fixed admissible constants, $p_0=\Omega(1)$ and this factor disappears.
\end{fact}



\begin{theorem}
\label{thm:hard-cap-parameterized}
For parameters satisfying conditions~\eqref{eq:hc-beta}--\eqref{eq:hc-open-slack},
Algorithm~\ref{alg:open-recolor} ($\OpenRecolor$)
returns a valid coloring
with at most
$\beta(\Delta+1)=\Delta+1+\ell$
colors
and frequency range $[\sigma/2,\lceil\alpha\sigma\rceil]$.
Its 
running times are $\Tseq(\OpenRecolor)=O(m(\lg\Delta+\lg n/p_0))$, $\Tcongest(\OpenRecolor)=O\!\left(\lg\Delta\cdot(D+\Delta+\lg^5\lg n)+(D+\Delta)\frac{\lg n}{p_0}\right)$ and 
$\Tcc(\OpenRecolor)=O(\lg\Delta+\lg n/p_0)$.  If $\alpha$
is a fixed constant,
these simplify to
$\Tseq(\OpenRecolor)=O(m\lg n)$, $\Tcongest(\OpenRecolor)=O((D+\Delta)\lg n+\lg\Delta\cdot\lg^5\lg n)$ and $\Tcc(\OpenRecolor)=O(\lg n)$, with high probability.
\end{theorem}

\begin{proof}
The call to $\ILDRC$ 
returns frequencies in $[\sigma/2,\sigma]$.  If Algorithm $\OpenRecolor$ returns immediately, the theorem follows from $\alpha\ge1$.
Otherwise, the algorithm keeps exactly $\beta(\Delta+1)$ target colors.
These colors had frequency at least $\sigma/2$ before recoloring and only
gain vertices, so the lower bound is preserved. Since they start with frequency
at most $\sigma\le\alpha\sigma\le\Flarge$, Lemma ~\ref{lem:open-cap-safety}
gives validity and the upper bound $\Flarge=\lceil\alpha\sigma\rceil$.

As observed above, every history reached by Algorithm~\ref{alg:open-recolor} ($\OpenRecolor$) satisfies
$E_{\mathrm{GoodHist}}(t)$.  Thus Lemma~\ref{lem:open-cap-progress}
applies in every round: conditioned on the history, each still uncolored
vertex is colored with probability at least $p_0/4$. Iterating this
conditional bound, for any fixed vertex $v$ and any
$T\ge 4(c+1)\ln n/p_0$, the probability that $v$ remains uncolored is
at most $(1-p_0/4)^T\le n^{-(c+1)}$. 
A union bound over all vertices gives
termination after $O(\lg n/p_0)$ recoloring iterations with probability at least $1-n^{-c}$. 
For fixed admissible $\alpha,\beta,\varphi$, the definition of $p_0$
has only positive constant terms, hence $p_0=\Omega(1)$.
The runtime statement follows from Fact~\ref{lem:open-recolor-runtime-interface} and is calculated in Section~\ref{sec:congest-model}.
\end{proof}

We now consider what happens when $\ell$ tends to zero.
Let $\epsilon>0$ be a constant, let $\alpha\ge1$, and let $\varphi$ satisfy
\begin{eqnarray}
& 1<1+2\epsilon \le 2\alpha\Delta/(\Delta+1),
\label{eq:hc-cor-r}
\\
& \max\left\{\frac12,\frac1{1+2\epsilon}\right\}<\varphi<1.
\label{eq:hc-cor-phi}
\end{eqnarray}
Set $\rho(\alpha)=(1+2\epsilon)/(2\alpha)$ and
$\ell_{\epsilon}=\lceil\rho(\alpha)(\Delta+1)\rceil$.

\begin{corollary}
\label{cor:hard-cap-exact-tradeoff}
Algorithm~\ref{alg:open-recolor} ($\OpenRecolor$), run with
$\ell=\ell_{\epsilon}$ and $\beta=1+\ell_{\epsilon}/(\Delta+1)$, returns a valid coloring with frequency range $[\sigma/2, \lceil\alpha\sigma\rceil]$ and palette size $\numclr \le (1+\rho(\alpha))(\Delta+1)+1$.
The running times are
$\Tseq(\OpenRecolor)=O(m\lg n)$,
$\Tcongest(\OpenRecolor)=O((D+\Delta)\lg n+\lg\Delta\cdot\lg^5\lg n)$ and
$\Tcc(\OpenRecolor)=O(\lg n)$, with high probability.
\end{corollary}

\begin{proof}
By Eq.~\eqref{eq:hc-cor-r}, $\ell_{\epsilon}\le\Delta$, so
Condition~\eqref{eq:hc-beta} holds for
$\beta=1+\ell_{\epsilon}/(\Delta+1)$.  Conditions~\eqref{eq:hc-alpha} and
\eqref{eq:hc-phi} hold by the assumptions of the corollary.  Moreover, 
\begin{equation}
\label{eq: 2 varphi alpha (beta-1)}
2\varphi\alpha(\beta-1)
=
2\varphi\alpha\cdot\frac{\ell_{\epsilon}}{\Delta+1}
\ge
2\varphi\alpha\cdot\rho(\alpha)
=
\varphi(1+2\epsilon)
>
1,
\end{equation}
where the last inequality follows from Eq.~\eqref{eq:hc-cor-phi}.  Thus
Condition~\eqref{eq:hc-open-slack} also holds, and
Theorem~\ref{thm:hard-cap-parameterized} applies.

The palette bound follows from
\[
\numclr
\le
\Delta+1+\ell_{\epsilon}
\le
(\Delta+1)+\rho(\alpha)(\Delta+1)+1
=
(1+\rho(\alpha))(\Delta+1)+1 .
\]
It remains only to note that $p_0=\Omega(1)$ for these constant parameters.
By Eq.~\eqref{eq: delta>0} and~\eqref{eq: 2 varphi alpha (beta-1)}, 
\[
\delta
=
\frac{2\varphi\alpha(\beta-1)-1}{2\varphi\alpha-1}
\ge
\frac{\varphi(1+2\epsilon)-1}{2\varphi\alpha-1}
>0 ,
\]
and the right-hand side is a positive constant.  Since
$1-\beta/2\le1/2$, the third term in the definition of $p_0$ is also bounded
below by a positive constant.  Hence $p_0=\Omega(1)$.
The running times are obtained from Fact~\ref{lem:open-recolor-runtime-interface} by the substitutions carried out in Section~\ref{sec:congest-model}.
\end{proof}

By way of a concrete example, fix a constant $0<\epsilon<1/2$, consider graphs satisfying $1+2\epsilon\le 2\Delta/(\Delta+1)$, and fix a constant
$\varphi$ satisfying $1/(1+2\epsilon)<\varphi<1$. Then we get the following.

\begin{corollary}
\label{cor:hard-cap-unit}
Algorithm~\ref{alg:open-recolor} ($\OpenRecolor$), run with
$\alpha=1$ and $\ell=\lceil(1/2+\epsilon)(\Delta+1)\rceil$, returns a valid coloring with at most
$(1.5+\epsilon)(\Delta+1)+1$ colors and frequency range
$[\sigma/2,\lceil\sigma\rceil]$. Its running times are $\Tseq(\OpenRecolor)=O(m\lg n)$, $\Tcongest(\OpenRecolor)=O((D+\Delta)\lg n+\lg\Delta\cdot\lg^5\lg n)$ and $\Tcc(\OpenRecolor)=O(\lg n)$, with high probability.
\end{corollary}




\section{$\Delta+1$ Palette Reduction with a Logarithmic frequency threshold}
\label{sec:delta-plus-two-log-cap}


Section~\ref{sec:hard-capped-open-reduction} gave a capped recoloring
procedure that reduces the palette to
$\beta(\Delta+1)=\Delta+1+\ell$ colors.  Its progress condition,
Eq.~\eqref{eq:hc-open-slack}, requires
$2\varphi\alpha\ell/(\Delta+1)>1$.  Hence pushing that one-shot reduction
all the way to $\Delta+1$ colors would 
force
$\alpha=\Omega(\Delta)$ and give no useful near-equitable bound.

This section handles the $\numclr=\Delta+1$ endpoint by applying the same capped
recoloring idea gradually.  Starting from the balanced coloring produced by
$\ILDRC$, the algorithm repeatedly deletes a constant fraction of the current
excess colors and recolors the deleted vertices into the remaining palette.
In each phase, every surviving color (that is still in $\cP$) may receive only $O(\sigma)$ additional
vertices.  Since the excess decreases geometrically, the number of phases is
$O(\lg\Delta)$.  The resulting price for reaching $\Delta+1$ colors is
therefore an upper frequency threshold of order $\sigma\lg\Delta$, rather
than an $O(\sigma)$ threshold.

For fixed parameters $\varepsilon$ and $c$, set $C_\varepsilon = 100/\varepsilon$.
Throughout this section, assume the requirements
\begin{equation}
\label{eq:sigma-condition-log algorithm}
0<\varepsilon\le 1,\qquad
c\ge 1,\qquad
\sigma ~\ge~ \max\{ 6/\varepsilon,\ C_\varepsilon (c+6)\lg n\}.
\end{equation}
Let
\[
\Flarge^{ph}=\lceil(1+\varepsilon)\sigma\rceil
\qquad\mbox{and}\qquad
\Flarge
=
\sigma+\lceil\lg (\Delta+1)\rceil\cdot\Flarge^{ph}.
\]
The main result of the section is that Algorithm~\ref{alg:log-compact} ($\LogCompact$)
computes, with probability at least $1-n^{-c}$ 
a valid coloring with at most
$\Delta+1$ colors in $O(\lg n\lg\Delta)$ $\CC$ rounds.  The lower frequency threshold remains $\sigma/2$, and the upper frequency threshold is
$\Flarge = O(\lg\Delta \cdot\sigma)$.

\subsection{Overview and Goal}
\label{subsec:log-cap-overview}

Algorithm~\ref{alg:log-compact} ($\LogCompact$) starts by applying the capped recoloring idea of Section~\ref{sec:hard-capped-open-reduction} in phases, to achieve $\numclr=\Delta+2$.  In one phase,
it removes the $\numclr_{remove}= \lfloor\numclr_{excess}/2\rfloor$ least frequent colors and recolors their vertices into the remaining target palette by calling
Algorithm~\ref{alg:budget-recolor} ($\BudgetRecolor$).

The call to $\BudgetRecolor$ enforces an \emph{individual phase} budget rather than a \emph{global} frequency cap.  At the beginning of the phase, each target color has counter $\Fadd(i)=0$; during the phase it may color at most $\Flarge^{ph}=\lceil(1+\varepsilon)\sigma\rceil$ new vertices.  A target color $i$
is \emph{open} in a recoloring iteration $t$ only if at the beginning of the iteration, $\Fadd(i) \le \Fcap= \lfloor(1+\varepsilon/2)\sigma\rfloor$. Note that if $i$ is open in iteration $t$, it might accept more than $\Fcap-\Fadd(i)$ new candidates, making it closed for the remaining iterations of the current phase, but the algorithm enforces that it will not ``overflow'' beyond the absolute threshold $\Flarge^{ph}$ in the current iteration.
We refer to $B_i = \Flarge^{ph} -\Fadd(i)$ as the budget of color $i$ for the remainder of the current phase.
The gap of $\Flarge^{ph}-\Fcap= \Omega(\varepsilon\sigma)$ unused colors is used in the concentration argument.

Since each phase removes a constant fraction of the current excess palette of $\numclr-(\Delta+1)$ colors, and $\numclr=O(\Delta)$,
there are $O(\lg\Delta)$ phases.  Each ``surviving'' color (that remains in $\cP$) receives only $O(\sigma)$ new vertices per phase. It follows that the upper frequency threshold is only $\Flarge = O(\lg\Delta \cdot\sigma)$. 

The algorithm concludes with applying one final step of $\Recolor$, discarding the smallest frequency color and ending with an optimal palette size of $\Delta+1$.

The formal descriptions of $\LogCompact$ and $\BudgetRecolor$ appear in Algorithms~\ref{alg:log-compact}
and~\ref{alg:budget-recolor}.

\begin{algorithm}[ht!]
\begin{center}
\begingroup
\begin{minipage}{\linewidth}
\footnotesize
\newcount\LogCompactLineNo
\LogCompactLineNo=0
\newcommand{\LogCompactLine}[2]{%
  \global\advance\LogCompactLineNo by 1%
  \par\noindent
  \makebox[1.8em][r]{\scriptsize\the\LogCompactLineNo}\hspace{0.45em}%
  \parbox[t]{\dimexpr\linewidth-2.25em\relax}{\raggedright\hangindent=#1\hangafter=1\noindent\hspace*{#1}#2}%
  \par}
\newcounter{BudgetLineNo}
\setcounter{BudgetLineNo}{0}
\newcommand{\BudgetLine}[2]{%
  \refstepcounter{BudgetLineNo}%
  \par\noindent
  \makebox[1.8em][r]{\scriptsize\theBudgetLineNo}\hspace{0.45em}%
  \parbox[t]{\dimexpr\linewidth-2.25em\relax}{\raggedright\hangindent=#1\hangafter=1\noindent\hspace*{#1}#2}%
  \par}
\begin{tabular}{@{}p{0.455\linewidth}@{\hspace{0.045\linewidth}}p{0.455\linewidth}@{}}
\begin{minipage}[t]{\linewidth}
\refstepcounter{algocf}
\label{alg:log-compact}
\hrule
\vspace{2pt}
\noindent{\bfseries Algorithm~\thealgocf:} $\LogCompact$ 
\vspace{2pt}
\hrule
\vspace{2pt}
\raggedright
\LogCompactLine{0pt}{\textbf{Input:} graph $G=(V,E)$ with maximum degree $\Delta$, constants $0<\varepsilon\le1$ and $c\ge1$.}
\LogCompactLine{0pt}{\textbf{Output:} a valid coloring with at most $\Delta+1$ colors and frequency range $[\sigma/2, 2\sigma+\lceil\lg(\Delta+1)\rceil\lceil(1+\varepsilon)\sigma\rceil]$.}
\LogCompactLine{0pt}{ $\clr\gets\ILDRC(G)$
{\scriptsize ($\numclr\le2(\Delta+1)$ and $\freq(i)\in[\sigma/2,\sigma]$)}}
\LogCompactLine{0pt}{\textbf{while} $\numclr>\Delta+2$ \textbf{do}}
\LogCompactLine{1.2em}{$\numclr_{excess}\gets\numclr-(\Delta+1)$.}
\LogCompactLine{1.2em}{$\numclr_{remove}\gets\lfloor\numclr_{excess}/2\rfloor$.}
\LogCompactLine{1.2em}{Sort colors by nondecreasing frequency: $\freq(i_1)\le\cdots\le\freq(i_{\numclr})$.}
\LogCompactLine{1.2em}{$\cP_{remove}\gets\{i_1,\ldots,i_{\numclr_{remove}}\}$.}
\LogCompactLine{1.2em}{$V_{target}\gets\bigcup_{i\in\cP_{remove}}V(i)$.}
\LogCompactLine{1.2em}{$\cP_{target}\gets\cP\setminus\cP_{remove}$.}
\LogCompactLine{1.2em}{Cancel the colors of all vertices in $V_{target}$.}
\LogCompactLine{1.2em}{$\Flarge^{ph}\gets\lceil(1+\varepsilon)\sigma\rceil$, $\Fcap\gets\lfloor(1+\varepsilon/2)\sigma\rfloor$, $p_0\gets\varepsilon^2/16$.}

\LogCompactLine{1.2em}{
Run $\BudgetRecolor$ on $G$, $\clr$, $V_{target}$, $\cP_{target}$ with parameters $\Flarge^{ph}$, $\Fcap$, $p_0$.}
\LogCompactLine{1.2em}{Set $\clr$ to the output of this call.}
\LogCompactLine{0pt}{\textbf{end while} ($\numclr=\Delta+2$)}
\LogCompactLine{0pt}{Sort current colors by nonincreasing frequency:
$\freq(i_1)\ge\cdots\ge\freq(i_{\numclr})$.}
\LogCompactLine{0pt}{$\cP_{target}\gets\{i_1,\ldots,i_{\Delta+1}\}$.}
\LogCompactLine{0pt}{$V_{target}\gets V_{i_\numclr}$.}
\LogCompactLine{0pt}{Cancel the color $i_\numclr$.}
\LogCompactLine{0pt}{$\clr \gets \Recolor(V_{target},\cP_{target})$}
\LogCompactLine{0pt}{\textbf{return} $\clr$.}

\end{minipage}
&
\begin{minipage}[t]{\linewidth}
\refstepcounter{algocf}
\label{alg:budget-recolor}
\hrule
\vspace{2pt}
\noindent{\bfseries Algorithm~\thealgocf:} $\BudgetRecolor$ 
\vspace{2pt}
\hrule
\vspace{2pt}
\raggedright
\BudgetLine{0pt}{\textbf{Input:} graph $G=(V,E)$, target vertices $V_{target}$, current partial coloring $\clr$, target palette $\cP_{target}$.}
\BudgetLine{0pt}{\textbf{Parameters:} thresholds $\Flarge^{ph}$ and $\Fcap$, activation probability $p_0$.}
\BudgetLine{0pt}{\textbf{Output:} a valid extension of $\clr$ to $V_{target}$.}
\BudgetLine{0pt}{For every $i\in\cP_{target}$ set $\Fadd(i)\gets0$.}
\BudgetLine{0pt}{\textbf{while} $V_{target}\neq\emptyset$ \textbf{do}}
\BudgetLine{1.2em}{$\cP_{open}\gets\{i\in\cP_{target}\mid \Fadd(i)\le \Fcap\}$.}
\BudgetLine{1.2em}{\textbf{for each} $v\in V_{target}$ in parallel \textbf{do}}
\BudgetLine{2.4em}{$\cP_{neig}(v)\gets\{\clr(u)\mid u\in N(v)\setminus V_{target}\}$.
\label{step: Pneig(v)}}
\BudgetLine{2.4em}{$\cP_{target}(v)\gets\cP_{target}\setminus\cP_{neig}(v)$.}
\BudgetLine{2.4em}{$\cP_{open}(v)\gets\cP_{target}(v)\cap\cP_{open}$.
\label{step: Popen(v)}}
\BudgetLine{2.4em}{Become active with probability $p_0$.}
\BudgetLine{2.4em}{\textbf{if} $v$ is active and $\cP_{open}(v)\neq\emptyset$ \textbf{then}}
\BudgetLine{3.6em}{Pick $c_v\in\cP_{open}(v)$ uniformly at random and broadcast $\textsf{Propose}(c_v)$.}
\BudgetLine{3.6em}{If no neighbor in $V_{target}$ proposes $c_v$, then $v$ becomes a candidate for color $c_v$.}
\BudgetLine{1.2em}{\textbf{for each color} $i\in\cP_{target}$ in parallel \textbf{do}}
\BudgetLine{2.4em}{Let $C_i$ be the set of candidates for $i$.}
\BudgetLine{2.4em}{Draw a fresh random ranking for the vertices of $C_i$.}
\BudgetLine{2.4em}{Let $B_i\gets\Flarge^{ph}-\Fadd(i)$.}
\BudgetLine{2.4em}{$C_i^{add}\gets$ the $\min\{|C_i|,B_i\}$ highest-ranked candidates.}
\BudgetLine{2.4em}{For every $v\in C_i^{add}$, set $\clr(v)\gets i$.}
\BudgetLine{2.4em}{$\Fadd(i)\gets\Fadd(i)+|C_i^{add}|$.}
\BudgetLine{1.2em}{$V_{target}\gets V_{target}\setminus\bigcup_{i\in\cP_{target}} C_i^{add}$.}
\BudgetLine{0pt}{\textbf{return} $\clr$.}
\end{minipage}
\end{tabular}
\vspace{2pt}
\end{minipage}
\endgroup
\end{center}
\end{algorithm}


\subsection{Analysis}
\label{subsec:log-cap-analysis}


\noindent 
The next lemma is the phase-budget analogue of Lemma~\ref{lem:open-cap-safety}; the proof is the same deterministic safety argument, with the per-phase counter $\Fadd(i)$ replacing the global frequency $\freq(i)$. 

\begin{lemma}
\label{lem:budget-recolor-safety}
\leavevmode\par\noindent
At every point in Algorithm~\ref{alg:budget-recolor} ($\BudgetRecolor$), the partial coloring is valid and every target color $i\in\cP_{target}$ satisfies $\Fadd(i)\le \Flarge^{ph}$.  Consequently, if the procedure terminates, it returns a valid extension of the input coloring and each target color receives at most $\Flarge^{ph}=\lceil(1+\varepsilon)\sigma\rceil$ new vertices during the phase.
\end{lemma}

\begin{proof}
A vertex proposes only a color not used by any already colored neighbor. If two adjacent uncolored vertices propose the same color, then neither of them becomes a candidate.
Hence, accepting a candidate preserves validity.

For the frequency threshold, fix a color $i$ and an iteration of Algorithm~\ref{alg:budget-recolor} ($\BudgetRecolor$).  The algorithm allows at most
$\Flarge^{ph}-\Fadd(i)$ candidates to use the color $i$, so after the iteration, $\Fadd(i)\le \Flarge^{ph}$.  The claim follows by induction over the iterations.
\end{proof}

\begin{lemma}
\label{lem:log-cap-dropped-volume}
In a phase of Algorithm~\ref{alg:log-compact} ($\LogCompact$), let
$\numclr_{remove}=\lfloor\numclr_{excess}/2\rfloor$.Then
$|V_{target}| = \tfreq(\cP_{remove}) \le \numclr_{remove}\cdot\sigma$.
\end{lemma}

\begin{proof}
The average frequency over the current palette satisfies $n/\numclr\le n/(\Delta+1)=\sigma$.  Since $\cP_{remove}$ consists of the
$\numclr_{remove}$ least frequent colors, its average frequency is at most the
global average.  Therefore $\tfreq(\cP_{remove})\le \numclr_{remove}\cdot\sigma$.
\end{proof}


\noindent The following lemma uses the same open/closed palette counting template as Lemma~\ref{lem:many-open-choices}; here the phase parameter $s$ and the counters $\Fadd(i)$ replace the fixed target-palette parameters of Section~\ref{sec:hard-capped-open-reduction}.
For an execution of Procedure $\BudgetRecolor$, let $V_{target}^t$ be the set $V_{target}$ at the beginning of iteration $t$, and let $\du^t(v)$ be the number of neighbors of $v$ that are still uncolored at the beginning of iteration $t$.
Let $\cP_{neig}^t(v)$ and $\cP_{open}^t(v)$ be the sets $\cP_{neig}(v)$ and $\cP_{open}(v)$ constructed for the uncolored vertex $v$ in Steps {\ref{step: Pneig(v)}} and {\ref{step: Popen(v)}} of iteration $t$, respectively.
Let $\cP_{closed}^t =\cP_{target} \setminus \cP_{open}^t$.

\begin{lemma}
\label{lem:budget-open-choices}
Consider a call to Algorithm~\ref{alg:budget-recolor} ($\BudgetRecolor$) with $\numclr_{target}=\Delta+1+s$ and $|V_{target}|\le s\sigma$ for some integer $s\ge1$.
Assume Condition \eqref{eq:sigma-condition-log algorithm} holds.  In every iteration $t$, $\cP_{open}^t(v)$ 
satisfies $\numclr_{open}^t(v) \ge \du^t(v)+1+\varepsilon s/4$. 
\end{lemma}


\begin{proof}
At any iteration $t$, the set $\cP_{neig}^t(v)$ has size 
$\numclr_{neig}^t(v) \le \Delta-\du^t(v)$, so by Eq. {\eqref{eq: def Ptarget}},
\begin{equation}
\label{eq: Xtarget lb alg log}
\numclr_{target}^t(v)
\ge \numclr_{target} - \numclr_{neig}^t(v)
\ge
(\Delta+1+s)-(\Delta-\du^t(v))
=
\du^t(v)+s+1 .
\end{equation}
A color $i$ that is closed at the beginning of iteration $t$ satisfies
$\Fadd(i)>\Fcap$.  Since
$\Fcap=\lfloor(1+\varepsilon/2)\sigma\rfloor$ and $\sigma\ge6/\varepsilon$ by Eq. {\eqref{eq:sigma-condition-log algorithm}}, we have
$\Fcap\ge(1+\varepsilon/3)\sigma$, so at the beginning of iteration $t$, each closed color $i$ has $\Fadd(i)>(1+\varepsilon/3)\sigma$.  

Note that the vertices that were colored during the phase all come from $V_{target}^1$, so at the beginning of iteration $t$,
$$
\mbox{$\numclr_{closed}^t\cdot (1+\varepsilon/3)\sigma < \sum_{i\in\cP_{closed}^t} \Fadd(i) \le |V_{target}^1| \le s\sigma$.}
$$
Hence, the number of currently closed colors satisfies
\begin{equation}
\label{eq: Xclosed ub alg log}
\numclr_{closed}^t \le \frac{s\sigma}{(1+\varepsilon/3)\sigma} \le \left(1-\frac{\varepsilon}{4}\right)s.
\end{equation}
The last inequality follows since
$(1-\varepsilon/4)(1+\varepsilon/3)
=1+\varepsilon(1-\varepsilon)/12\ge1$ for $0<\varepsilon\le1$.
By Eq. {\eqref{eq: Xtarget lb alg log}} and {\eqref{eq: Xclosed ub alg log}},
$$\numclr_{open}^t(v) \ge \numclr_{target}(v)-\numclr_{closed}^t \ge \du^t(v)+1+ \varepsilon s/4.
\qedhere $$
\end{proof}


\noindent The conditioning convention below is the same as the one used before Lemma~\ref{lem:open-cap-progress}; only the phase-specific bad events and the resulting one-iteration progress estimate change.

For the next lemma, fix one call to Algorithm~\ref{alg:budget-recolor} ($\BudgetRecolor$).  For a recoloring iteration $t$, let
$\HIST_t$ denote the history of the execution up to the beginning of iteration
$t$.  Let $E_{\mathrm{GoodHist}}(t)$ be the event that no bad event occurred
in iterations $1,\ldots,t-1$, where a bad event means that (a) the partial
coloring became invalid, (b) some target color $i$ satisfied
$\Fadd(i)>\Flarge^{ph}$, or (c) some uncolored vertex $v$ had fewer than
$\du^r(v)+1+\varepsilon s/4$ open locally valid colors at the beginning of some earlier iteration $r<t$.

We observe that any history $\HIST_t$ reached by Algorithm~\ref{alg:budget-recolor} ($\BudgetRecolor$) by the end of iteration $t-1$ satisfies $E_{\mathrm{GoodHist}}(t)$.
This is verified by induction on $t$.  For $t=1$, the input partial coloring is valid, all counters
$\Fadd(i)$ are zero, and Lemma~\ref{lem:budget-open-choices} gives the open-choice bound.  For the inductive step, Lemma~\ref{lem:budget-recolor-safety}
preserves validity and the bound $\Fadd(i)\le \Flarge^{ph}$, and
Lemma~\ref{lem:budget-open-choices} gives the open-choice bound for the next
history.

\noindent In the next proof, we condition on a fixed history $\HIST_t$ satisfying
$E_{\mathrm{GoodHist}}(t)$.  Under this conditioning, the partial coloring,
the uncolored set, the counters $\Fadd(i)$, and all palettes are fixed;
the only remaining randomness is the fresh random choices made in iteration $t$.
For readability, we omit this conditioning from the notation and write simply
$\mathbb P[E']$ and $\mathbb E[X]$ instead of
$\mathbb P[E'\mid E_{\mathrm{GoodHist}}(t)]$ and
$\mathbb E[X\mid E_{\mathrm{GoodHist}}(t)]$.

\begin{lemma}
\label{lem:budget-recolor}
Fix $0<\varepsilon\le1$, $C_\varepsilon=100/\varepsilon$ and
$c\ge 1$ satisfying Condition Eq. {\eqref{eq:sigma-condition-log algorithm}}.
\leavevmode\par\noindent
Suppose Algorithm~\ref{alg:budget-recolor} ($\BudgetRecolor$) is called with $\numclr_{target}=\Delta+1+s$ for an integer $s\ge1$, $|V_{target}|\le s\sigma$, $\Flarge^{ph}= \lceil(1+\varepsilon)\sigma\rceil$, $\Fcap=\lfloor(1+\varepsilon/2)\sigma\rfloor$, and $p_0=\varepsilon^2/16$.  Then, with probability at least
$1-n^{-(c+2)}$, Algorithm~\ref{alg:budget-recolor} ($\BudgetRecolor$) terminates in $O(\varepsilon^{-2}(c+4)\lg n)$ iterations, it returns a valid coloring for $G$ and each color $i\in \cP_{target}$ receives at most $\Flarge^{ph}$ new vertices during the phase.
\end{lemma}

\begin{proof}
Let $T_0=\left\lceil 64\varepsilon^{-2}(c+4)\lg n\right\rceil$. 
We prove that, with probability at least $1-n^{-(c+2)}$, all vertices of
$V_{target}$ are colored during the first $T_0$ iterations.  The validity of the
coloring and the bound on the number of vertices added to each target color then
follow from Lemma~\ref{lem:budget-recolor-safety}.

We first control the event that too many vertices propose the same open color in
one iteration.  Fix an iteration $t\le T_0$, condition on a history
$\HIST_t$ satisfying $E_{\mathrm{GoodHist}}(t)$, and fix an open color
$i\in\cP_{open}$.  Let $Y_{i,t}$ be the number of active vertices that propose
color $i$ in iteration $t$, before conflict resolution.  By
Lemma~\ref{lem:budget-open-choices}, every vertex that can propose $i$ has at
least $1+\varepsilon s/4$ open locally valid choices, and therefore proposes
$i$ with probability at most $p_0/(1+\varepsilon s/4)$.  Since
$|V_{target}|\le s\sigma$,
\[
\mu:=\mathbb E[Y_{i,t}]
\le
\frac{p_0}{1+\varepsilon s/4}\cdot |V_{target}|
\le
\frac{p_0s\sigma}{1+\varepsilon s/4}
\le
\frac{4p_0}{\varepsilon}\sigma
=
\frac{\varepsilon}{4}\sigma .
\]
On the other hand, since $i$ is open, $\Fadd(i)\le\Fcap$, and hence its remaining
budget satisfies
\[
B_i=\Flarge^{ph}-\Fadd(i)
\ge
\Flarge^{ph}-\Fcap
\ge
\frac{\varepsilon}{2}\sigma .
\]
Thus $B_i\ge 2\mu$.  If $\mu=0$, then $Y_{i,t}=0$ deterministically.  Otherwise,
applying the Chernoff bound to the threshold $B_i$ gives
\begin{equation}
\label{eq: Yit probability}
\mathbb P[Y_{i,t}>B_i]
\le
\exp(-\varepsilon\sigma/12)
\le
n^{-(c+8)},
\end{equation}
where the last inequality follows from
Condition~\eqref{eq:sigma-condition-log algorithm} and the choice
$C_\varepsilon=100/\varepsilon$.

Let $E_{\mathrm{Safe}}$ be the event that, in every one of the first $T_0$
iterations and for every open color $i$, the number of proposals to $i$ is at
most its remaining budget $B_i$.  We now apply the union bound to the estimate of Eq. \eqref{eq: Yit probability} over
all colors and all iterations.  In the calls made by Algorithm~\ref{alg:log-compact},
the number of colors is at most $2(\Delta+1)\le2n$.  Also, since
$\sigma\le n$, Condition~\eqref{eq:sigma-condition-log algorithm} implies
$\varepsilon^{-1}\le n/6$ and $\varepsilon^{-1}(c+6)\lg n\le n/100$.
Consequently, $T_0=\lceil64\varepsilon^{-2}(c+4)\lg n\rceil\le n^2$ for $n\ge2$.  Hence
\[
\mathbb P[\overline{E_{\mathrm{Safe}}}]
\le
2n\cdot n^2\cdot n^{-(c+8)}
\le
n^{-(c+4)} .
\]
Next, to prove progress, we bound the probability that an uncolored vertex $v$ becomes a conflict-free candidate in a given iteration.
Fix a vertex $v$ that is still uncolored at the
beginning of some iteration $t\le T_0$, and condition on the history up to that
iteration.  By Lemma~\ref{lem:budget-open-choices}, $v$ has at least
$\du^t(v)+1$ open locally valid colors.  If $v$ becomes active, it chooses one
of these colors uniformly.  For any uncolored neighbor $u$ of $v$, the probability that $u$ proposes the same color as $v$ is at most
$p_0/\numclr_{open}^t(v)$.  Taking a union bound over the $\du^t(v)$ uncolored
neighbors of $v$, the probability that some such neighbor proposes the same color is at most
\[
\frac{p_0}{\numclr_{open}^t(v)} \cdot \du^t(v)
\le
\frac{\du^t(v)}{\du^t(v)+1} \cdot p_0
\le
p_0 .
\]
Therefore, in every iteration $t$ in which $v$ is still uncolored, the probability of the event $E_{success}(v,t)$ that $v$ becomes a conflict-free candidate 
satisfies
\begin{equation}
\label{eq: prob success}
\mathbb{P}[E_{success}(v,t)] ~\ge~ p_0(1-p_0) ~\ge~ p_0/2.
\end{equation}
On the event $E_{\mathrm{Safe}}$, every conflict-free candidate is accepted:
indeed, for each color $i$, the total number of proposals to $i$ is at most
$B_i$, so the number of candidates for $i$ is also at most $B_i$, and the quota
step admits all of them.  Hence, if $E_{\mathrm{Safe}}$ holds and a vertex $v$ ever becomes a conflict-free candidate during the first $T_0$ iterations, then $v$ is colored.

It remains to bound the probability of the event $E_{never}(v) \equiv \bigcap_{t=1}^{T_0} {\bar E}_{success}(v,t)]$ that the vertex $v$ \emph{never} becomes a conflict-free candidate.  The 
lower bound of Eq. \eqref{eq: prob success} holds conditionally on every history in which the vertex is still uncolored, so by multiplying the conditional failure probabilities,
\[
\mathbb P[E_{never}(v)
]
\le
(1-p_0/2)^{T_0}
\le
\exp(-p_0T_0/2)
\le
n^{-2(c+4)} .
\]
Taking a union bound over all vertices gives probability at most
$n\cdot n^{-2(c+4)}\le n^{-(c+3)}$ that some vertex remains uncolored after
$T_0$ iterations while $E_{\mathrm{Safe}}$ holds.  Combining this with
$\mathbb P[\overline{E_{\mathrm{Safe}}}]\le n^{-(c+4)}$, the total failure
probability is at most $n^{-(c+2)}$ for $n\ge2$.

Thus, with probability at least $1-n^{-(c+2)}$, Algorithm $\BudgetRecolor$
terminates within $T_0=O(\varepsilon^{-2}(c+4)\lg n)$ iterations.  By
Lemma~\ref{lem:budget-recolor-safety}, the returned coloring is valid and every
target color receives at most $\Flarge^{ph}$ new vertices.
\end{proof}

\begin{lemma}
\label{lem:logcompact-final-cleanup}
Let $\Fpre = \sigma+\lceil\lg(\Delta+1)\rceil\lceil(1+\varepsilon)\sigma\rceil.$ Suppose that after the while loop of Algorithm~\ref{alg:log-compact} ($\LogCompact$), the current coloring is valid, has 
$\Delta+2$ colors, and every color has frequency in $[\sigma/2,\Fpre]$.
Then the final cleanup step of Algorithm~\ref{alg:log-compact}, if its call to $\Recolor$ succeeds, returns a valid coloring with 
$\Delta+1$ colors and frequency range
$[\sigma/2, \Fpre+\sigma].$
\end{lemma}

\begin{proof}
After the while loop terminates, Algorithm~\ref{alg:log-compact} keeps the $\Delta+1$ most frequent colors as $\cP_{target}$ and cancels 
the least frequent color, $i_\numclr$, where $\numclr=\Delta+2$.  The set of canceled vertices is $V_{target} = V(i_\numclr)$.

Since the frequency of the least frequent 
color 
is at most the average 
frequency,
\[
|V_{target}| \le \frac{n}{\Delta+2} < \frac{n}{\Delta+1}=\sigma .
\]

The target palette has size $\Delta+1$, so Lemma~\ref{lemma:partial_extension} applies to the partial coloring obtained after canceling $V_{target}$.  Hence the call to $\Recolor(V_{target},\cP_{target})$ returns a valid coloring using only colors of $\cP_{target}$.  Therefore the final palette size is at most $\Delta+1$.

Turning to the resulting frequencies,
every color in $\cP_{target}$ already existed before the cleanup and had frequency at least $\sigma/2$; during the cleanup it could only gain vertices. Hence it still preserves the lower frequency threshold $\sigma/2$. 
The coloring before the cleanup obeyed
the upper frequency threshold 
$\Fpre$. During the cleanup, any fixed target color can gain at most all vertices of $V_{target}$, and $|V_{target}|\le\sigma$.  Therefore the output colors obey the upper frequency threshold $\Fpre+\sigma =
2\sigma+\lceil\lg(\Delta+1)\rceil\lceil(1+\varepsilon)\sigma\rceil$.
The lemma follows.
\end{proof}

For the sake of the time analysis of Section~\ref{sec:congest-model} we note the following:

\begin{fact}
\label{lem:logcompact-runtime-interface}
Algorithm $\LogCompact$ begins with one call to $\ILDRC$. The excess palette size drops by a factor of two in each phase, so the algorithm has $O(\lg\Delta)$ phases. Lemma~\ref{lem:budget-recolor} gives $O(\varepsilon^{-2}(c+4)\lg n)$ iterations of $\BudgetRecolor$ per phase. Each such iteration uses $\CA$ for counters and open-color information, and $\QA$ for quota admission. Then, the algorithm calls once to $\Recolor$. When $c, \varepsilon$ are fixed admissible constants, $\varepsilon^{-2}(c+4)=\Omega(1)$ and this factor disappears.
\end{fact}



\begin{theorem}
\label{thm:delta-plus-two-log-cap}
For constants $c$ and $\varepsilon$ satisfying Conditions \eqref{eq:sigma-condition-log algorithm},
Algorithm~\ref{alg:log-compact} ($\LogCompact$) computes, with probability at least
$1-n^{-c}$, a valid coloring with at most $\Delta+1$ colors and frequency range $[\sigma/2, 2\sigma+
\lceil\lg (\Delta+1)\rceil
\lceil(1+\varepsilon)\sigma\rceil]$. Its running times are $\Tseq(\LogCompact) = O(m\lg n\lg\Delta)$, $\Tcongest(\LogCompact) = O((D+\Delta)\lg n\lg\Delta)$ and $\Tcc(\LogCompact)=O(\lg n\lg\Delta)$. Every output color $i$ satisfies
\end{theorem}

\begin{proof}
The lower bound is preserved throughout the algorithm.  The initialization by
$\ILDRC$ gives $\freq(i)\ge\sigma/2$ for every color.  Thereafter, colors are
either removed or gain vertices; no new color is created.

Consider one phase, and let
$\numclr_{remove}=\lfloor\numclr_{excess}/2\rfloor$.  The target palette size is
\[
    \numclr_{target}
    =
    \numclr-\numclr_{remove}
    =
    \Delta+1+s,
    \qquad
    s=\numclr_{excess}-\numclr_{remove}
      =\lceil\numclr_{excess}/2\rceil .
\]
\noindent By Lemma~\ref{lem:log-cap-dropped-volume}, $|V_{target}|\le \numclr_{remove}\cdot\sigma\le s\sigma$.  Therefore Lemma~\ref{lem:budget-recolor} applies to the call to Algorithm~\ref{alg:budget-recolor} ($\BudgetRecolor$) in this phase.  During the phase, each target color receives at most $\lceil(1+\varepsilon)\sigma\rceil$ new vertices.

The excess satisfies
\[
    \numclr_{excess}^{\,new}
    =
    \numclr_{excess}
    -
    \lfloor\numclr_{excess}/2\rfloor
    =
    \lceil\numclr_{excess}/2\rceil .
\]
Since initially $\numclr_{excess}\le\Delta+1$, after at most
$\lceil\lg (\Delta+1)\rceil$ phases the excess is at most $1$, which is
equivalent to $\numclr\le\Delta+2$.



A surviving color starts with frequency at most $\sigma$ after
$\ILDRC$ and can gain at most $\lceil(1+\varepsilon)\sigma\rceil$ vertices in
each phase.  This gives the stated upper bound $\Fpre$. Then, combined with Lemma~\ref{lem:logcompact-final-cleanup}, we get our needed results.

Finally, use the standard high-probability guarantee of $\ILDRC$ with failure
probability at most $n^{-(c+1)}$, and use Lemma~\ref{lem:budget-recolor} for
the recoloring phases, in addition to the standard high-probability guarantee of $\Recolor$.  Each recoloring phase fails with probability
at most $n^{-(c+2)}$.  There are at most
$\lceil\lg (\Delta+1)\rceil\le \lg  n+1\le n$ such phases for $n\ge2$, so
the total recoloring failure probability is at most $n^{-(c+1)}$.  A union
bound over the initialization and all recoloring phases gives failure
probability at most $2n^{-(c+1)}\le n^{-c}$, as claimed.
The runtime statement follows from Fact~\ref{lem:logcompact-runtime-interface} and is calculated in Section~\ref{sec:congest-model}.
\end{proof}

\section{Implementation Details and Runtimes}
\label{sec:congest-model}
\label{sec:implementation-details-runtimes}
This section supplies the model-dependent costs for the primitives of Section~\ref{sec:tools} and then substitutes them into the local runtime facts appearing in the algorithmic sections.  

\noindent{\bf 
Costs in the sequential model:} 


As mentioned earlier, the modular description of our algorithms using basic primitives does not exactly fit the sequential implementation.
For example, in the distributed implementations, before invoking $\LC$, the vertices must learn the globally chosen target palette, or equivalently the target-list information induced by that palette; we charge this preparation as a $\CA$-type operation. In the sequential implementation this bookkeeping need not be performed as a separate primitive: the algorithm stores the current color classes and can form the admissible lists while scanning the graph for the sequential $\LC$ simulation. This distinction does not affect any of our asymptotic bounds, since both the sequential bookkeeping and the sequential $\LC$ simulation cost $O(m)$. Thus, even when a higher-level algorithm invokes $\Recolor$ repeatedly, charging each activation schematically as target-list preparation plus one $\LC$ call gives the same sequential running time up to constant factors.

Let us now discuss the sequential implementation of our algorithms in more detail.
The graph is stored by adjacency lists together with the current color-class lists.  The $\LC$ primitive is implemented directly on these lists.  For each target vertex, the algorithm marks the colors already used by its colored neighbors, scans the implicit admissible palette until it finds an unmarked color, assigns that color, and then clears the temporary marks.  Equivalently, this is the standard sequential simulation of list coloring when the lists are represented by a global palette together with locally forbidden neighbor colors.  Each edge incident to the target set is inspected only a constant number of times, so the $\LC$ call costs $O(m)$ in our notation.  $\CA$ is implemented by a linear scan that increments per-color counters, maintains the requested per-color records, and bucket-sorts colors when an order by frequency is needed.  $\QA$ is implemented by scanning the candidate lists for the participating colors, keeping the best ranked or first feasible candidates up to the prescribed quota.  Hence 
$
\Tseq(\LC)=\Tseq(\CA)=\Tseq(\QA)=O(m).
$
Here $\CA$ recomputes or bucket-sorts color records in linear time, and $\QA$ scans candidate lists.  Thus 
$
\Tseq(\DRC)=\Tseq(\Recolor)=\Tseq(\Split)=O(m).
$

\noindent{\bf 
Costs in the $\CONGEST$ model:}
The global coordination primitives are routed on a BFS tree.  A $\CA$ call sends one $O(\lg n)$-bit record per relevant color through a pipelined convergecast, aggregates the counts or other per-color statistics at the root, and broadcasts the resulting color records back to the vertices.  Since all palettes used by the algorithms have size $O(\Delta)$ after the initial coloring stage, the pipeline cost is $O(D+\Delta)$.  A $\QA$ call is similar: candidate vertices send color-indexed requests or ranks, the coordinator/root admits at most the requested quota for each color, and the decisions are pipelined back down the tree, again using $O(\Delta)$ color records.  For $\LC$, once the local lists are known, we use the randomized $\CONGEST$ list-coloring algorithm of Halldorsson, Nolin, and Tonoyan~\cite{HalldorssonNolinTonoyan2022}, getting
$\Tcongest(\LC)=O(\lg^5\lg n)$.
%
%
For the global primitives we get
$\Tcongest(\CA)=\Tcongest(\QA)=O(D+\Delta).$

In the Congested Clique, every pair of vertices can exchange an $O(\lg n)$-bit message in one round, and standard routing lets the color-indexed records used by $\CA$ and $\QA$ be delivered to their coordinators in constant rounds.  Each color can be assigned a coordinator, which receives the relevant frequency counts, candidate ranks, or quota requests, makes the deterministic admission/counting decision locally, and sends the answer back to the affected vertices.  The $\LC$ primitive is constant-round 
by the Congested Clique coloring algorithms of Czumaj, Davies, and Parter~\cite{CzumajDaviesParter2020,CzumajDaviesParter2021}.  Therefore $\Tcc(\LC)=\Tcc(\CA)=\Tcc(\QA)=O(1)$.  The named procedures are then charged through these primitives.  $\DRC$ and $\Recolor$ are $\LC$ calls, with the relevant target lists prepared by the surrounding $\CA$ step when needed.  $\Split$ is charged as $\CA$ plus $\QA$, not as a separate primitive: $\CA$ identifies the colors above the threshold and computes their descendant quotas, while $\QA$ assigns the vertices of each such color to final descendant colors.

\noindent{\bf Algorithm calculations.}
Substituting the primitive costs stated above into Facts~\ref{lem:nbc-runtime-interface}, \ref{lem:ildrc-runtime-interface}, \ref{lem:pftradeoff-runtime-interface}, \ref{lem:pftradeoff-small-runtime-interface}, \ref{lem:pttradeoff-runtime-interface}, \ref{lem:pttradeoff-small-runtime-interface}, \ref{lem:open-recolor-runtime-interface}, and~\ref{lem:logcompact-runtime-interface} gives the time bounds in our theorems and Table~\ref{tbl:results summary}.
\section{Future Directions}

%



In this work, we presented a suite of algorithms for near-equitable coloring and discussed their implementation in the $\CC$, congest and sequential models. Several promising avenues for future research arise from our results and the limitations of the current approach.

\noindent{\bf 
Optimized Recoloring Strategies.}
The core of our iterative algorithms relies on Procedure $\Recolor$, where low frequency color are dissolved and the uncolored vertices select new colors uniformly at random from their available palette. 
This method may not be optimal for convergence speed or load balancing. A potential improvement would be to introduce non-uniform probability distributions. For instance, a vertex could choose a target color $c_j$ with a probability inversely proportional to $|V(j)|$, thereby naturally biasing the process towards filling smaller under-capacity classes first. Additionally, rather than attempting to recolor with probability 1 (which causes high contention), vertices could retain their current color with a certain probability, or attempt to switch only if local contention is below a threshold. Exploring these weighted or probabilistic heuristics could yield faster convergence rates or tighter bounds on the frequency range.


\noindent{\bf 
Topology-Specific Optimizations.}
While our results apply to general graphs, many distributed systems operate on specific topologies such as rings, trees, or grids. It is well known that for rings and trees, a 3-coloring can be computed in $O(\lg^* n)$ rounds in the $\LOCAL$ model. An interesting question is whether equitable coloring can be achieved with similar efficiency on special graph classes. Specifically, can we achieve a near-equitable coloring on a ring or tree in $O(\lg^* n)$ rounds, or does the global balance constraint impose a strictly higher lower bound on the runtime, regardless of the local topology?


\noindent{\bf 
Optimal deterministic solution in $\CC$}
A natural and plausible 
direction  is
to exploit the powerful techniques of Czumaj, Davies, and Parter ~\cite{CzumajDaviesParter2020,CzumajDaviesParter2021} in order to derive a deterministic constant-round algorithm for near-equitable coloring in the $\CC$ model, that for any constant parameters $\rho, \eta \in (0,1)$ computes a valid coloring $\clr: V \to \cP$ with palette size $\numclr \le (1+\rho)(\Delta+1)$ and frequency range $[\max\{1, \lfloor(1-\eta)\sigma\rfloor\}, \lceil(1+\eta)\sigma\rceil]$. Such an algorithm would require adapting the recursive partitioning technique, parallel seed-fixing machinery and bounded-independence concentration framework of 
\cite{CzumajDaviesParter2020,CzumajDaviesParter2021} and proving the balanced split estimates directly. To achieve near-equitable frequencies rather than arbitrary list colorings, it may be necessary to replace their palette invariants with a strict two-sided bin-size condition. Palettes can be exclusively established at the leaves of the recursion, and vertices experiencing unfavorable hash distributions would have to be postponed and absorbed bottom-up subject to strict capacity limits. We are currently working 
on making these ideas rigorous and writing up the resulting algorithm and analysis.

Note that it appears that the resulting coloring will require $(1+\rho)(\Delta+1)$ colors. A more challenging goal would be to derive a fast algorithm of this type that uses $\Delta+o(\Delta)$ colors, or even a fast algorithm comparable to Algorithm $\LogCompact$ of Section \ref{sec:delta-plus-two-log-cap}, which uses $\Delta+1$ colors. 
Another challenge is to find a viable way to extend this approach to the $\CONGEST$ or sequential models.


\clearpage




\begin{thebibliography}{10}

\bibitem{BakerCoffman1996mutual}
B.~S. Baker and J.~Coffman, Edward~G.
\newblock Mutual exclusion scheduling.
\newblock {\em Theoretical Computer Science}, 162(2):225--243, 1996.

\bibitem{CzumajDaviesParter2020}
A.~Czumaj, P.~Davies, and M.~Parter.
\newblock Simple, deterministic, constant-round coloring in the congested
  clique.
\newblock In {\em Proceedings of the 39th ACM Symposium on Principles of
  Distributed Computing (PODC)}, pages 310--318, 2020.

\bibitem{CzumajDaviesParter2021}
A.~Czumaj, P.~Davies, and M.~Parter.
\newblock Simple, deterministic, constant-round coloring in congested clique
  and {MPC}.
\newblock {\em SIAM Journal on Computing}, 50(5):1603--1626, 2021.

\bibitem{DeWerra1985timetabling}
D.~de~Werra.
\newblock An introduction to timetabling.
\newblock {\em European Journal of Operational Research}, 19(2):151--162, 1985.

\bibitem{HajnalSzemeredi1970}
A.~Hajnal and E.~Szemer{\'e}di.
\newblock Proof of a conjecture of {P}. erd{\H{o}}s.
\newblock In P.~Erd{\H{o}}s, A.~R{\'e}nyi, and V.~T. S{\'o}s, editors, {\em
  Combinatorial Theory and Its Applications, Vol. II}, pages 601--623.
  North-Holland, Amsterdam, 1970.

\bibitem{HalldorssonNolinTonoyan2022}
M.~M. Halld{\'o}rsson, A.~Nolin, and T.~Tonoyan.
\newblock Overcoming congestion in distributed coloring.
\newblock In {\em Proceedings of the 2022 ACM Symposium on Principles of
  Distributed Computing (PODC)}, pages 26--36. Association for Computing
  Machinery, 2022.

\bibitem{Johansson-IPL-99}
{\"O}.~Johansson.
\newblock Simple distributed $\delta+1$-coloring of graphs.
\newblock {\em Information Processing Letters}, 70:229--–232, 1999.

\bibitem{Kierstead2010fast}
H.~A. Kierstead, A.~V. Kostochka, M.~Mydlarz, and E.~Szemer{\'e}di.
\newblock A fast algorithm for equitable coloring.
\newblock {\em Combinatorica}, 30(2):217--224, 2010.

\bibitem{LPPP03}
Z.~Lotker, E.~Pavlov, B.~Patt-Shamir, and D.~Peleg.
\newblock {MST} construction in {O}(log log n) communication rounds.
\newblock In {\em Proc. 15th ACM Symp. on Parallelism in Algorithms and
  Architectures (SPAA)}, pages 94–--100, 2003.

\bibitem{Peleg00:book}
D.~Peleg.
\newblock {\em Distributed Computing: A Locality-Sensitive Approach}.
\newblock SIAM, 2000.

\bibitem{Wood1969timetabling}
D.~C. Wood.
\newblock A technique for colouring a graph applicable to large scale
  timetabling problems.
\newblock {\em The Computer Journal}, 12(4):317--319, 1969.

\end{thebibliography}
\end{document}